\documentclass[11pt]{amsart}
\usepackage{fullpage}
\usepackage[foot]{amsaddr}
\usepackage{microtype}
\usepackage[OT1]{fontenc}
\usepackage{eulervm}
\usepackage[tt=false]{libertine} 
\usepackage{amsmath}
\usepackage{amssymb}
\usepackage{amsthm}
\usepackage{mathtools} 
\usepackage[linesnumbered,boxed,ruled,vlined]{algorithm2e}
\usepackage{algpseudocode}
\usepackage{enumitem}
\usepackage{bm}
\usepackage{bbm}
\usepackage{xifthen}

\usepackage[margin=1cm]{caption} 
\usepackage{subcaption}

\usepackage[thinlines]{easytable}

\usepackage[bookmarks=true,hypertexnames=false,pagebackref]{hyperref}
\hypersetup{colorlinks=true, citecolor=blue, linkcolor=red, urlcolor=blue}

\usepackage{pgfplots}
\pgfplotsset{compat=1.16}
\usepackage{tikz}
\usetikzlibrary{arrows,arrows.meta,backgrounds,calc,fit,decorations.pathreplacing,decorations.markings,shapes.geometric}

\tikzstyle{internal} = [draw, fill, shape=circle]
\tikzstyle{external} = [shape=circle]
\tikzstyle{square}   = [draw, fill, rectangle]
\tikzstyle{triangle} = [draw, fill, regular polygon, regular polygon sides=3, inner sep=3pt]
\tikzstyle{pentagon} = [draw, fill, regular polygon, regular polygon sides=5, inner sep=2pt, minimum size=14pt]
\tikzset{every fit/.append style=text badly centered}

\usetikzlibrary{positioning,chains,fit,shapes,calc}
\usetikzlibrary{trees}
\usetikzlibrary{decorations.pathreplacing}
\usetikzlibrary{decorations.pathmorphing}
\usetikzlibrary{decorations.markings}
\tikzset{>=latex} 

\usepackage{ifthen}

\usepackage{cleveref}

\usepackage[textsize=tiny]{todonotes}

\usepackage[normalem]{ulem}

\usepackage{mleftright}

\usepackage{cool}
\Style{DSymb={\mathrm d},DShorten=true,IntegrateDifferentialDSymb=\mathrm{d}}

\newcommand{\tp}[1]{{\left( #1 \right)}}

\newcommand{\Ex}{\mathop{\mathbb{{}E}}\nolimits}
\renewcommand{\Pr}{\mathop{\mathrm{Pr}}\nolimits}

\def\*#1{\mathbf{#1}}
\def\+#1{\mathcal{#1}}
\def\-#1{\mathrm{#1}}
\def\=#1{\mathbb{#1}}

\newcommand{\abs}[1]{\ensuremath{\left\vert#1\right\vert}}

\newcommand{\eps}{\varepsilon}

\newcommand{\Var}[2]{\ensuremath{\textnormal{Var}_{#1}\left(#2\right)}}

\newcommand{\defeq}{:=}

\newcommand{\numP}{\#{\textnormal{\textbf{P}}}}

\newcommand{\BIS}{\#\textnormal{\textbf{BIS}}}

\newcommand{\DTV}[2]{d_{\text{TV}}\left({#1},{#2}\right)}

\newtheorem{theorem}{Theorem}

\newtheorem{lemma}[theorem]{Lemma}

\newtheorem{observation}[theorem]{Observation}

\newtheorem{corollary}[theorem]{Corollary}
\theoremstyle{definition}

\newtheorem{definition}[theorem]{Definition}
\newtheorem{fact}[theorem]{Fact}

\theoremstyle{remark}
\newtheorem{remark}[theorem]{Remark}

\crefname{theorem}{Theorem}{Theorems}
\crefname{observation}{Observation}{Observations}
\crefname{claim}{Claim}{Claims}
\crefname{condition}{Condition}{Conditions}
\crefname{algorithm}{Algorithm}{Algorithms}
\crefname{property}{Property}{Properties}
\crefname{example}{Example}{Examples}
\crefname{fact}{Fact}{Facts}
\crefname{lemma}{Lemma}{Lemmas}
\crefname{corollary}{Corollary}{Corollaries}
\crefname{definition}{Definition}{Definitions}
\crefname{remark}{Remark}{Remarks}
\crefname{proposition}{Proposition}{Propositions}
\crefname{equation}{equation}{equations}
\crefname{enumi}{Case}{Case}
\creflabelformat{enumi}{(#2#1#3)}

\makeatletter
\def\prob#1#2#3{\goodbreak\begin{list}{}{\labelwidth\z@ \itemindent-\leftmargin
      \itemsep\z@  \topsep6\p@\@plus6\p@
      \let\makelabel\descriptionlabel}
  \item[\textbf{Name}]#1
  \item[\textbf{Instance}]#2
  \item[\textbf{Output}]#3
  \end{list}}
\makeatother

\makeatletter
\providecommand\@dotsep{5}
\def\listtodoname{Todo list}
\def\listoftodos{\@starttoc{tdo}\listtodoname}
\makeatother

\newcommand{\rch}[3]{{#1}\leadsto_{_#3}{#2}}
\newcommand{\nrch}[3]{{#1}\not\leadsto_{_{#3}}{#2}}

\newboolean{doubleblind}
\setboolean{doubleblind}{false}

\title{An FPRAS for two terminal reliability in directed acyclic graphs}

\ifdoubleblind
\author{Author(s)}
\else

\author{Weiming Feng and Heng Guo}

\address[Weiming Feng]{Institute for Theoretical Studies, ETH Z\"urich, Clausiusstrasse 47, 8092 Z\"urich, Switzerland.}
\address[Heng Guo]{School of Informatics, University of Edinburgh, Informatics Forum, Edinburgh, EH8 9AB, United Kingdom.}
\email{weiming.feng@eth-its.ethz.ch}
\email{hguo@inf.ed.ac.uk}

\fi

\begin{document}

\begin{abstract}
  We give a fully polynomial-time randomized approximation scheme (FPRAS) for two terminal reliability in directed acyclic graphs (DAGs).
  In contrast, we also show the complementing problem of approximating two terminal unreliability in DAGs is \BIS{}-hard.
\end{abstract}
\maketitle
\section{Introduction}\label{section-def}

Network reliability is one of the first problems studied in counting complexity.
Indeed, $s-t$ reliability is listed as one of the first thirteen complete problems when Valiant \cite{Val79a} introduced the counting complexity class \numP.
The general setting is that given a (directed or undirected) graph $G$,
each edge $e$ of $G$ fails independently with probability $q_e$.
The problem of $s-t$ reliability is then asking the probability that in the remaining graph, the source vertex $s$ can reach the sink $t$.
There are also other variants, where one may ask the probability of various kinds of connectivity properties of the remaining graph.
These problems have been extensively studied, and apparently most variants are $\numP$-complete \cite{Ball80,Jer81,BP83,PB83,Ball86,Col87}.

While the exact complexity of reliability problems is quite well understood,
their approximation complexity is not.
Indeed, the approximation complexity of the first studied $s-t$ reliability is still open in either directed or undirected graphs.
One main exception is the \emph{all-terminal} version (where one is interested in the remaining graph being connected or disconnected).
A famous result by Karger \cite{Kar99} gives the first fully polynomial-time randomized approximation scheme (FPRAS) for all-terminal \emph{unreliability},
while about two decades later, Guo and Jerrum \cite{GJ19a} give the first FPRAS for all-terminal \emph{reliability}.
The latter algorithm is under the partial rejection sampling framework \cite{GJL19},
and the Markov chain Monte Carlo (MCMC) method is also shown to be efficient shortly after \cite{ALOV19,CGM19}.
See \cite{HS18,Kar20,CHLP23}, \cite{GH20}, and \cite{ALOVV21,CGZZ23} for more recent results and improved running times along the three lines above for the all-terminal version respectively.

The success of these methods implies that the solution space of all-terminal reliability is well-connected via local moves.
However, this is not the case for the two-terminal version (namely the $s-t$ version). 
Instead, the natural local-move Markov chain for $s-t$ reliability is torpidly mixing.
Here the solution space consists of all spanning subgraphs (namely a subset of edges) in which $s$ can reach $t$.
Consider a (directed or undirected) graph composed of two paths of equal length connecting $s$ and $t$.
Suppose we start from one path and leave the other path empty.
Then before the other path is all included in the current state, we cannot remove any edge of the initial path.
This creates an exponential bottleneck for local-move Markov chains, and it suggests that a different approach is required.

In this paper, we give an FPRAS for the $s-t$ reliability in directed acyclic graphs.
Note that the exact version of this problem is \numP-complete \cite[Sec 3]{PB83},
even restricted to planar DAGs where the vertex degrees are at most $3$ \cite[Theorem 3]{Pro86}.
Our result positively resolves an open problem by Zenklusen and Laumanns \cite{ZL11}.
Without loss of generality, in the theorem below we assume that any vertex other than $s$ has at least one incoming edge,
and thus $\abs{E}\ge \abs{V}-1$ for the input $G=(V,E)$.

\begin{theorem}  \label{thm:main}
  Let $G=(V,E)$ be a directed acyclic graph (DAG), failure probabilities $\*q=(q_e)_{e \in E} \in [0,1]^E$, two vertices $s,t\in V$, and $\eps>0$.
  There is a randomized algorithm that takes $(G,\*q,s,t,\eps)$ as inputs and outputs a $(1\pm\eps)$-approximation to the $s-t$ reliability with probability at least $3/4$ in time $\widetilde{O}(n^{6}m^4\max\{m^4,\eps^{-4}\})$ where $n=\abs{V}$, $m=\abs{E}$, and $\widetilde{O}$ hides $\text{polylog}(n/\epsilon)$ factors.
\end{theorem}
The running time of \Cref{thm:main} is $\widetilde{O}(n^6m^8)$ when $\eps>1/m$, and is $\widetilde{O}(n^{6}m^4/\eps^4)$ when $\eps<1/m$.
The reason behind this running time is that our algorithm always outputs at least a $(1\pm 1/m)$-approximation.
Thus, when $\eps> 1/m$, it does not matter what $\eps$ actually is for the running time.
This high level of precision is required for the correctness of the algorithm.


As hinted earlier, our method is a significant departure from the techniques for the all-terminal versions.
Indeed, a classical result by Karp and Luby \cite{KL83,KLM89} has shown how to efficiently estimate the size of a union of sets.
A direct application of this method to $s-t$ reliability is efficient only for certain special cases \cite{KL85,ZL11}.
Our main observation is to use the Karp-Luby method as a subroutine in dynamic programming using the structure of DAGs.
Let $s=v_1,\dots,v_n=t$ be a topological ordering of the DAG~$G$. (Note that we can ignore vertices before $s$ and after $t$.)
Let $R_u$ be the $u-t$ reliability so that our goal is to estimate $R_s$.
We inductively estimate $R_{v_i}$ from $i=n$ to $i=1$.
For each vertex $u$, our algorithm maintains an estimator of $R_u$ and a set $S_u$ of samples of subgraphs in which $u$ can reach $t$.
In the induction step, we use the Karp-Luby method to generate the next estimator,
and use a self-reduction similar to the Jerrum-Valiant-Vazirani sampling to counting reduction \cite{JVV86} to generate samples.
Both tasks can be done efficiently using only what have been computed for previous vertices.
Moreover, a naive implementation of this outline would require exponentially many samples.
To avoid this, we reuse generated samples and carefully analyze the impact of doing so on the overall error bound.
A more detailed overview is given in \Cref{sec:alg-overview}.

Our technique is inspired by an FPRAS for the number of accepting strings of non-deterministic finite automata (\#NFA), found by Arenas, Croquevielle, Jayaram, and Riveros~\cite{ACJR21}, which in turn used some techniques from a quasi-polynomial-time algorithm for sampling words from context-free languages by Gore, Jerrum, Kannan, Sweedyk, and Mahaney \cite{GJKSM97}.
Their \#NFA algorithm runs in time $O\left(\left(\frac{n\ell}{\eps}\right)^{17}\right)$,\footnote{The running time of the algorithm in \cite{ACJR21} is not explicitly given. This bound is obtained by going through their proof.} where $n$ is the number of states and $\ell$ is the string length.
More recently, Meel, Chakraborty, and Mathur claim an improved running algorithm which runs in time~$\widetilde{O}(\frac{n^4\ell^{11}}{\eps^4})$~\cite[Theorem 3]{MCM23}.
These algorithms first normalize the NFA into a particular layered structure. 
Applying similar methods on the $s-t$ reliability problem can simplify the analysis, but would greatly slow down the algorithm.
In contrast, our method works directly on the DAG. 
This makes our estimation and sampling subroutines interlock in an intricate way.
To analyze the algorithm, we have to carefully separate out various sources of randomness.
This leads to a considerably more sophisticated analysis, with a reward of a much better (albeit still high) running time.

Independently and simultaneously, Amarilli, van Bremen, and Meel~\cite{AvBM23} also found an FPRAS for $s-t$ reliability in DAGs.
Their method is to reduce the problem to \#NFA via a sequence of reductions, and then invoke the algorithm in \cite{ACJR21}.
Indeed,
\ifdoubleblind 
\else as Marcelo Arenas subsequently pointed out to us,~\fi 
counting the number of subgraphs of a DAG in which $s$ can reach $t$ belongs to a complexity class $\*{SpanL}$~\cite{AJ93}, where \#NFA is $\*{SpanL}$-complete under polynomial-time parsimonious reductions. 
In particular, every problem in $\*{SpanL}$ admits an FPRAS because \#NFA admits one~\cite{ACJR21}, which implies that $s-t$ reliability in DAGs admits an FPRAS if $q_e = 1/2$ for all edges.
The method of \cite{AvBM23} reduces a reliability instance of $n$ vertices and $m$ edges, where $q_e = 1/2$ for all edges, to estimating length $m$ accepting strings of an NFA with $O(m^2)$ states.\footnote{In fact, \cite{AvBM23} first reduces the reliability instance to an nOBDD (non-deterministic ordered binary decision diagram) of size $O(m)$, which can be further reduced to an NFA of size $O(m^2)$. As they are working with a more general context, no explicit reduction is given for the $s-t$ reliability problem in DAGs. We provide a direct (and essentially the same) reduction in \Cref{sec:reduction}.}
As a consequence, their algorithm (even using the faster algorithm for \#NFA \cite{MCM23}) has a running time of $O\left(m^{19}\eps^{-4}\right)$.
When $q_e\neq 1/2$, their reduction needs to expand the instance further to reduce to the $q_e=1/2$ case, slowing down the algorithm even more.
In contrast, our algorithm deals with all possible probabilities $0\leq q_e < 1$ in a unified way.
In any case, an algorithm via reductions is much slower than the direct algorithm in \Cref{thm:main}.

As both all-terminal reliability and unreliability in undirected graphs have FPRASes \cite{Kar99,GJ19a},
one may wonder if FPRAS exists for $s-t$ \emph{unreliability} in DAGs.
Here $s-t$ unreliability is the probability that $s$ cannot reach $t$ in the random subgraph.
In contrast to \Cref{thm:main}, we show that this problem is \BIS-hard,
where \BIS{} is the problem of counting independent sets in bipartite graphs,
whose approximation complexity is still open.
This is a central problem in the complexity of approximate counting \cite{DGGJ04},
and is conjectured to have no FPRAS.

\begin{theorem}  \label{thm:st-unrel}
  There is no FPRAS to estimate $s-t$ \emph{unreliability} in DAGs unless there is an FPRAS for \BIS{}.
  This is still true even if all edges fail with the same probability.
\end{theorem}

\Cref{thm:st-unrel} is proved in \Cref{sec:BIS-hard}.
The hardness of $s-t$ unreliability does not contradict \Cref{thm:main}.
This is because a good relative approximation of $x$ does not necessarily provide a good approximation of $1-x$, especially when $x$ is close to $1$.

The complexity of estimating $s-t$ reliability in general directed or undirected graphs remains open.
We hope that our work sheds some new light on these decades old problems.
Another open problem is to reduce the running time of \Cref{thm:main}, as currently the exponent of the polynomial is still high.

\subsection{Algorithm overview}
\label{sec:alg-overview}
Here we give an overview of our algorithm. 
For simplicity, we assume that~$q_e = 1/2$ for all edges.
The general case of $0 \leq q_e < 1$ can be solved with very small tweaks.

Let $s=v_1\prec\dots\prec v_n=t$ be a topological ordering of the DAG $G$. (Note that we can ignore vertices before $s$ and after $t$.)
Let $R_u$ be the $u-t$ reliability so that our goal is to estimate $R_s$.
Note that $R_{v_i}$ depends only on the subgraph induced by the vertex set $\{v_i,v_{i+1},\ldots,v_n\}$. We denote this subgraph by $G_{v_i}=(V_{v_i},E_{v_i})$.
As we assumed $q_e=1/2$, estimating $R_{v_i}$ is equivalent to estimating the number of (spanning) subgraphs of $G_{v_i}$ in which $v_i$ can reach $t$.
Denote the latter quantity by $Z_{v_i}$ so that $R_{v_i}=Z_{v_i}/2^{|E_{v_i}|}$.
For each vertex $v_i$, our algorithm maintains an estimator and a multi-set of random samples:
\begin{itemize}
  \item $\widetilde{Z}_{v_i}$:\footnote{Our algorithm actually directly maintains an estimate $\widetilde{R}_{v_i}$ to the reliability $R_{v_i}$. In this overview, we use $\widetilde{Z}_{v_i}$ instead for simplicity.}
      an estimator that approximates $Z_{v_i}$ with high probability;
  \item $S_{v_i}$: a multi-set of random subgraphs, where each $H \in S_{v_i}$ is an approximate sample from $\pi_{v_i}$ and $\pi_{v_i}$ is the uniform distribution over all spanning subgraphs of $G_{v_i}$ in which $v_i$ can reach $t$.
\end{itemize}

Our algorithm computes $\widetilde{Z}_{v_i}$ and $S_{v_i}$ for $i$ from $n$ to $1$ by dynamic programming. The base case $v_n = t$ is trivial. 
In the induction step, suppose the vertex $v_i$ has $d$ out-neighbors $u_1,u_2,\ldots,u_d$.
Note that each~$u_j$ for $j\in[d]$ comes after $v_i$ in the topological ordering. 
Let us further assume $v_i\prec u_1\prec,\dots,\prec u_d$.
For any subgraph $H$ of $G_{v_i}$, if $v_i$ can reach $t$ in $H$, then there exists a neighbor $u_j$ such that $v_i$ can reach $u_j$ and $u_j$ can reach $t$ in $H$.
We can write $Z_{v_i}$ as the size of the union $\Omega\defeq\cup_{j=1}^d \Omega^{(j)}$, where $\Omega^{(j)}$ contains all subgraph of $G_{v_i}$ where $v_i$ can reach $t$ through the neighbor $u_j$.
Note that it is straightforward to estimate the size of $\Omega^{(j)}$ given $\widetilde{Z}_{u_j}$,
and to generate uniform random subgraphs from $\Omega^{(j)}$ using samples in $S_{u_j}$. 
Given the size and samples from $\Omega^{(j)}$, a classical algorithm by Karp and Luby~\cite{KL83,KLM89} can be applied to efficiently estimate the size of the union of sets,
namely to compute $\widetilde{Z}_{v_i}$.

The more complicated task is to generate the samples of $S_{v_i}$.
We use a sampling to counting self-reduction \'a la Jerrum-Valiant-Vazirani~\cite{JVV86}. 
To generate a sample $H$, we go through each edge $e$ in $G_{v_i}$, deciding if $e$ is added into $H$ according to its conditional marginal probability. 
The first edge to consider is $(v_i,u_1)$.
Its marginal probability depends on how many subgraphs in $\Omega$ contain it.
This quantity is the same as the number of subgraphs in which $\Lambda$ can reach $t$, where $\Lambda$ is a new vertex after contracting $v_i$ and $u_1$.
To estimate this number, denoted by $Z_{\Lambda}$, we use the Karp-Luby algorithm again.
Note that the Karp-Luby algorithm requires estimates and samples from all out-neighbours of $\Lambda$,
which have been computed already as these vertices are larger in the topological ordering than $v_i$. 
Having estimated $Z_{\Lambda}$, we estimate the marginal probability of $(v_i,u_1)$ and decide if it is included in $H$.
The process then continues to consider the next edge.
In each step, we contract all vertices that can be reached from $v_i$ into $\Lambda$, and keep estimating new $Z_{\Lambda}$ using the Karp-Luby algorithm to compute the conditional marginal probabilities,
until all edges are considered to generate $H$.

A naive implementation of the process above is to generate fresh samples every time~$S_{u}$ is used,
which would lead to an exponential number of samples required.
To maintain efficiency, the key property of our algorithm is to \emph{reuse} random samples.
For any vertex $u$, the algorithm generates the multi-set of samples~$S_{u}$ only once.
Whenever the algorithm needs to use random samples for $u$, it always picks one sample from the same set~$S_{u}$.
Hence, one sample may be used multiple times during the whole algorithm.
Reusing samples introduces very complicated correlation among all $(\widetilde{Z}_u,S_u)$'s, which is a challenge to proving the correctness of the algorithm. 
Essentially, our analysis shows that as long as the estimates ($\widetilde{Z}_u$'s) and the samples are accurate enough, 
the overall error can be controlled.
Accurate estimates of~$Z_{u}$'s allow us to bound the total variation (TV) distance between our samples and perfect samples.
In turn, the small TV distance implies that there is a coupling between them, which helps us bound the errors of the estimates.
This way, we circumvent the effect of correlation on the analysis and achieve the desired overall error bound.

For the overall running time, there are two tasks for each vertex, namely the estimation step (based on Karp-Luby) and the sample generation step.
As we also need to perform estimation steps as subroutines when generating samples,
the running time is dominated by the time for the sampling step.
Let $\ell$ be the number of samples required per vertex, so that the total number of samples generated is $n\ell$.
Roughly speaking, because we can only use the union bound due to various correlation, and because errors accumulate throughout dynamic programming,
we set the error to be $\delta\defeq n^{-1}\min\{m^{-1},\eps\}$ for the estimation step.
Each estimation has two stages, first getting a constant approximation and then using the crude estimation to tune the parameters and obtain a $1\pm\delta$ approximation.
This succeeds with constant probability using $O(n\delta^{-2})$ samples.
The estimator itself requires $O(m)$ time to compute, and thus the total running time for each estimation step with constant success probability is~$\widetilde{O}(mn\delta^{-2})$.
However, as samples are reused, we need to apply a union bound to control the error over all possible values of the samples, which are exponentially many. 
This requires us to amplify the success probability of the estimation step to~$\exp(-\Omega(m))$, which means we need to repeat the algorithm $O(m)$ times and take the median.
Each estimation step thus takes $\widetilde{O}(m^2n\delta^{-2})$ time and $O(m n\delta^{-2})$ samples.
Instead of maintaining $O(m n\delta^{-2})$ samples for every vertex, 
the aforementioned two-stage estimation allows us to spread the cost and maintain~$\ell\defeq\frac{O(mn\delta^{-2})}{n}=O(n^2m\max\{m^2,\eps^{-2}\})$ samples per vertex,
which we show suffice with high probability.
As we may do up to $O(m)$ estimation steps during each sampling step,
the overall running time is then bounded by
  $\widetilde{O}(n\ell \cdot m \cdot m^2 n \delta^{-2}) = \widetilde{O}\tp{n^6m^4\max\{m^4,\eps^{-4}\}}$.

\section{Preliminaries}

\subsection{Problem definitions}\label{sec:prelim-def}

Let $G=(V,E)$ be a directed acyclic graph (DAG). Each directed edge (or arc) $e = (u,v)$ is associated with a failure probability $0\leq q_e<1$. (Any edge with $q_e=1$ can be simply removed.) 
We also assume graph $G$ is simple because parallel edges with failure probabilities $q_{e_1},q_{e_2},\ldots,q_{e_k}$ can be replace with one edge with failure probability $q_e = \prod_{i=1}^m q_{e_i}$.  
Given two vertices $s,t \in V$, the $s-t$ reliability problem asks the probability that $s$ can reach $t$ if each edge $e\in E$ fails (namely gets removed) independently with probability $q_e$.
Formally, let $\boldsymbol{q}=(q_e)_{e \in E}$.
The $s-t$ reliability problem is to compute
\begin{align}\label{eq-def-R}
    R_{G,\boldsymbol{q}}(s,t) \defeq \Pr_{\+G}[\,\text{there is a path from $s$ to $t$ in $\+G$}\,],
\end{align}
where $\+G =(V,\+E)$ is a random subgraph of $G=(V,E)$ such that each $e \in E$ is added independently to $\+E$ with probability $1-q_e$.

Closely connected to estimating $s-t$ reliability is a sampling problem, which we call the $s-t$ \emph{subgraph sampling} problem. Here the goal is to sample a random (spanning) subgraph $G'$ conditional on that there is at least one path from $s$ to $t$ in $G'$. Formally, let $\Omega_{G,s,t}$ be the set of all subgraphs $H = (V,E_H)$ of $G$ such that $E_H \subseteq E$ and $s$ can reach $t$ in $H$. The algorithm needs to draw samples from the distribution $\pi_{G,s,t,\boldsymbol{q}}$ whose support is $\Omega_{G,s,t}$ and for any $H =(V,E_H)\in \Omega_{G,s,t}$,
\begin{align}\label{eq-def-pi}
  \quad \pi_{G,s,t,\boldsymbol{q}}(H) = \frac{1}{R_{G,\boldsymbol{q}}(s,t)} \cdot \prod_{e \in E_H}(1-q_e)\prod_{f \in E\setminus E_H}q_f. 
\end{align}

\subsection{More notations}
Fix a DAG $G=(V,E)$. For any two vertices $u$ and $v$, we use
$\rch{u}{v}{G}$ to denote that $u$ can reach $v$ in the graph $G$ and use $\nrch{u}{v}{G}$ to denote that $u$ cannot reach $v$ in the graph $G$.
It always holds that $\rch{u}{u}{G}$.
Fix two vertices $s$ and $t$, where $s$ is the source and $t$ is the sink.
The failure probabilities $\*q$, $s$, and $t$ will be the same throughout the paper, and
thus we omit them from the subscripts.
For any vertex $u \in V$, we use $G_u = G[V_u]$ to denote the subgraph of $G$ induced by the vertex set
\begin{align}\label{eq-def-G_u}
  V_u \defeq \{w \in V \mid \rch{u}{w}{G} \,\land\,\rch{w}{t}{G} \}.
\end{align} 
Without loss of generality, we assume $G_s = G$. 
This means that all vertices except $s$ have at least one in-neighbour, and thus $m\ge n-1$.
If $G_s \neq G$, then all vertices and edges in $G - G_s$ have no effect on $s-t$ reliability and we can simply ignore them. For the sampling problem, we can first solve it on the graph $G_s$ and then independently add each edge $e$ in $G - G_s$ with probability $1 - q_e$.

Our algorithm actually solves the $u-t$ reliability and $u-t$ subgraph sampling problems in $G_u$ for all $u \in V$.
Let $G_u = (V_u,E_u)$. For any subgraph $H = (V_u, E_H)$ of $G_u$, define the weight function
\begin{align}\label{eqn:w_u}
	w_u(H) \defeq \begin{cases}
		\prod_{e \in E_H}(1-q_e)\prod_{f \in E_u \setminus E_H}q_f &\text{if } \rch{u}{t}{H};\\
		0 &\text{if } \nrch{u}{t}{H}.
	\end{cases}
\end{align} 
Define the distribution $\pi_u$ by
\begin{align*}
	\pi_u(H) \defeq \frac{w_u(H)}{R_u},
\end{align*}
where the partition function (the normalizing factor) 
\begin{align*}
	R_u \defeq \sum_{H: \text{ subgraph of $G_u$}}w_u(H)
\end{align*}
is exactly the $u-t$ reliability in the graph $G_u$. Finally, let
\begin{align}\label{eq-def-Omega_u}
	\Omega_{u} \defeq \left\{ H=(V_u,E_H) \mid E_H \subseteq E_u \,\land\, \rch{u}{t}{H} \right\}
\end{align}
be the support of $\pi_u$.
Also note that $R_s$ and $\pi_s$ are the probability $R_{G,\boldsymbol{q}}(s,t)$ and the distribution $\pi_{G,s,t,\boldsymbol{q}}$ defined in~\eqref{eq-def-R} and~\eqref{eq-def-pi}, respectively.
The set $\Omega_s$ is the set $\Omega_{G,s,t}$ in \Cref{sec:prelim-def}.

\subsection{The total variation distance and coupling}
Let $\mu$ and $\nu$ be two discrete distributions over $\Omega$. The \emph{total variation distance} between $\mu$ and $\nu$ is defined by
\begin{align*}
	\DTV{\mu}{\nu} \defeq \frac{1}{2}\sum_{x \in \Omega}\abs{\mu(x)- \nu(x)}.
\end{align*}
If $X \sim \mu$ and $Y \sim \nu$ are two random variables, we also abuse the notation and write $\DTV{X}{Y}\defeq\DTV{\mu}{\nu}$. 

A \emph{coupling} between $\mu$ and $\nu$ is a joint distribution $(X,Y)$ such that $X \sim \mu$ and $Y\sim \nu$. The following coupling inequality is well-known.
\begin{lemma}\label{lem:coupling}
For any coupling $\+C$ between two random variables $X\sim \mu$ and $Y\sim \nu$, it holds that
\begin{align*}
  \Pr_{\+C}[X \neq Y] \geq \DTV{\mu}{\nu}.
\end{align*}
Moreover, there exists an optimal coupling that achieves equality.
\end{lemma}

\section{The algorithm}

In this section we give our algorithm. 
We also give intuitions behind various design choices, and give some basic properties of the algorithm along the way.
The main analysis is in \Cref{sec-analysis}.

\subsection{The framework of the algorithm}

As $G=G_s$ is a DAG, there is a topological ordering of all vertices.
There may exist many topological orderings. We pick an arbitrary one, say, $v_1,\dots,v_n$.
It must hold that $v_1=s$ and $v_n=t$.
The topological ordering guarantees that if $(u,v)$ is an edge, $u$ must come before $v$ in the ordering, denoted $u\prec v$.

On a high level, our algorithm is to inductively compute an estimator $\widetilde{R}_u$ of $R_u$, from $u=v_n$ to $u=v_1$.
In addition to $\widetilde{R}_u$, we also maintain a multi-set $S_u$ of samples from $\pi_u$ over $\Omega_u$. 
For any vertex $u \in V$, let 
	$\Gamma_{\-{out}}(u) \defeq \{w \mid (u,w) \in E\}$
denote the set of out-neighbours of $u$.
Let
\begin{align}\label{eq-def-ell}
  \ell \defeq (60n + 150m)(400n + 500\left\lceil 10^4 n^2\max\{m^2,\eps^{-2}\}  \right\rceil) = O((n+m)n^2\max\{m^2,\epsilon^{-2}\})
\end{align}
be the size of $S_v$ for all $v\in V_u$, where $n$ is the number of vertices in $G$ and $m$ is the number of edges in $G$.
The choice of this parameter is explained in the last paragraph of \Cref{sec:alg-overview}.
Our algorithm is outlined in \Cref{alg-main}. Note that $G_t$ is an isolated vertex. 
For consistency, we let $S_t$ contain $\ell$ copies of $\emptyset$.
\begin{algorithm}  
  \caption{An FPRAS for $s-t$ reliabilities in DAGs} \label{alg-main}
  \KwIn{a DAG $G=(V,E)$, a vector $\boldsymbol{q}= (q_e)_{e \in E}$, the source $s$, the sink $t$, and an error bound $0 <\eps < 1$, where $G=G_s$ and $V=\{v_1 , v_2,\ldots,v_n\}$ is topologically ordered with $v_1= s$ and $v_n = t$}
  \KwOut{an estimator $\widetilde{R}_s$ of $R_s$}
  let $\tilde{R}_{t} = 1$ and $S_{t}$ be a multi set of $\ell$ $\emptyset$'s\label{line-base}\;
\For{$k$ from $n-1$ to $1$}{
	$\widetilde{R}_{v_k} \gets \textsf{ApproxCount}\tp{V_{v_k},E_{v_k}, v_k, (\widetilde{R}_{w}, S_w)_{w \in \Gamma_{\-{out}}(v_k)}}$\label{line-count}\;	
	$S_{v_k} \gets \emptyset$\;
	\For{$j$ from $1$ to $\ell$}{
	$S_{v_k} \gets S_{v_k} \cup \textsf{Sample}\tp{v_k,(\widetilde{R}_{w}, S_w)_{w \in \{v_{k+1},v_{k+2},\ldots,v_n\}}, \widetilde{R}_{v_k}}$\label{line-sample}\;
	}
} 
\Return $\widetilde{R}_s$.
\end{algorithm}

The base case (\Cref{line-base}) of $v_n=t$ is trivial.
The subroutine $\textsf{Sample}(\cdot)$ uses $(\widetilde{R}_{v_i},S_{v_i})$ for all $i>k$ and $\widetilde{R}_{v_k}$ to generate samples in $S_{v_k}$.
The subroutine $\textsf{ApproxCount}(V,E,u,(\widetilde{R}_w,S_w)_{\Gamma_{\-{out}}(u)})$ takes a graph $G=(V,E)$, a vertex $u$, and $(\widetilde{R}_w,S_w)$ for all $w \in \Gamma_{\-{out}}(u)$ as the input, and it outputs an approximation of the $u-t$ reliability in the graph $G$.
We describe $\textsf{Sample}$ in \Cref{sec:sample} and (a slightly more general version of) $\textsf{ApproxCount}$ in \Cref{sec:AC}.

\subsection{Generate samples} \label{sec:sample}
Let $u=v_k$ where $k<n$. Recall that $G_u = (V_u, E_u)$ is the graph defined in~\eqref{eq-def-G_u}. 
The sampling algorithm aims to output a random spanning subgraph $H=(V_u,\+E)$ from the distribution~$\pi_u$.
The algorithm is based on the sampling-to-counting reduction in~\cite{JVV86}.
It scans each edge $e$ in $E_u$ and decides whether to put $e$ into the set $\+E$ or not. The algorithm maintains two edge sets:
\begin{itemize}
	\item $E_1\subseteq E_u$: the set of edges that have been scanned by the algorithm;
    \item $\+E \subseteq E_1$: the current set of edges sampled by the algorithm.
\end{itemize}
Also, let $E_2 \defeq E_{u} \setminus E_1$ be the set of edges that have not been scanned yet by the algorithm.
Given any $\+E$, we can uniquely define the following subset of vertices
\begin{align}\label{eqn:Lambda-def}
  \Lambda = \Lambda_{\+E} \defeq \textsf{rch}(u,V_u,\+E) = \{w \in V_u \mid u \text{ can reach } w \text{ through edges in } \+E\}.
\end{align}
In other words, let $G'=(V_u,\+E)$ and $\Lambda$ is the set of vertices that $u$ can reach in $G'$. Note that $u \in \Lambda$ for any $\+E$. 
We will keep updating $\Lambda$ as $\+E$ expands.
When calculating the marginal probability of the next edge,
the path to $t$ can start from any vertex in $\Lambda$.
Thus we need to estimate the reliability from $\Lambda$ to $t$.
Instead of having a single source, as called in \Cref{alg-main},
we use a slightly more general version of \textsf{ApproxCount}, described in \Cref{sec:AC}, to allow a set $\Lambda$ of sources.
This subroutine $\textsf{ApproxCount}$ takes input $(V,E,\Lambda, (\widetilde{R}_w,S_w)_{w \in \partial \Lambda})$ and approximates the $\Lambda-t$ reliability in $(V,E)$,
which is the probability that there exists at least one vertex in $\Lambda$ being able to reach $t$ if each edge $e \in E$ fails independently with probability $q_e$.
An equivalent way of seeing it is to contract all vertices in $\Lambda$ into a single vertex $u$ first, and then calculate the $u-t$ reliability in the resulting graph.
\textsf{Sample} is described in \Cref{alg-sample},
and some illustration is given in \Cref{fig:sample-example}.

\begin{algorithm} 
  \caption{$\textsf{Sample}\tp{u,(\widetilde{R}_{w}, S_w)_{w \in \{v_{k+1},v_{k+2},\ldots,v_n\}},\widetilde{R}_{v_k}}$}\label{alg-sample} 
  \KwIn{a vertex $u=v_k$, all $(\widetilde{R}_{w}, S_w)$ for $w \in \{v_{k+1},v_{k+2},\ldots,v_n\}$, and $\widetilde{R}_{v_k}=\widetilde{R}_u$}
  \KwOut{a random subgraph $H=(V_u,\+E)$}
$T \gets\lceil 1000\log \frac{n}{\eps} \rceil$ and $F \gets 0$\;
$p_0 \gets \widetilde{R}_u$\label{line-p_0}\;
\Repeat
{$T < 0$ or $F = 1$}
{
	let $p \gets 1$\;
	let $E_1 \gets \emptyset$, $E_2 \gets E_u \setminus E_1$ and $\Lambda = \{u\}$\label{line-strat}\;
	\While{$t \notin \Lambda$}{
		let $\partial\Lambda \gets \{ w \notin \Lambda \mid \exists w' \in \Lambda \text{ s.t. } (w',w)\in E_2\}$; let $w^* \in \partial \Lambda$ be the smallest vertex in the topological ordering; pick an arbitrary edge $e = (w',w^*) \in E_2$ such that $w' \in \Lambda$\label{line-pick}\; 
		let $\Lambda_1 \gets \textsf{rch}(u,V_u,\+E \cup \{e\})$\;
        $E_1 \gets E_1 \cup \{e\}$ and $E_2 \gets E_2 \setminus \{e\}$\label{line-E2}\;
        $\partial \Lambda_1 \gets \{ w \notin \Lambda_1 \mid \exists w' \in \Lambda_1 \text{ s.t. } (w',w)\in E_2\}$ and $\partial \Lambda \gets \{ w \notin \Lambda \mid \exists w' \in \Lambda \text{ s.t. } (w',w)\in E_2\}$\label{line-Lambda-partial}\;
		$c_0 \gets \textsf{ApproxCount}(V_u, E_2, \Lambda, (\widetilde{R}_w,S_w)_{w \in \partial\Lambda} )$\label{line-c1}\;
		$c_1 \gets \textsf{ApproxCount}(V_u, E_2, \Lambda_1, (\widetilde{R}_w,S_w)_{w \in \partial \Lambda_1} )$\label{line-c2}\;
		let $c \gets 1$ with probability $\frac{(1-q_e)c_1}{q_ec_0 + (1-q_e)c_1}$; otherwise $c \gets 0$ \label{line-sample-c}\;
		\textbf{if} $c = 1$, \textbf{then} let $\+E \gets \+E \cup \{e\}$, $\Lambda \gets \Lambda_1$, $p\gets p \frac{(1-q_e)c_1}{q_ec_0 + (1-q_e)c_1}$\label{line-sample-filter}\;
		\textbf{if} $c = 0$, \textbf{then} let  $p\gets p \tp{1-\frac{(1-q_e)c_1}{q_ec_0 + (1-q_e)c_1}}$\;
		
}
\For{all edges $e \in E_2$}{
	let $c \gets 1$ with probability $1-q_e$; otherwise $c \gets 0$\;
	\textbf{if} $c = 1$, \textbf{then} let $\+E \gets \+E \cup \{e\}$ and $p \gets p(1-q_e)$\;
	\textbf{if} $c = 0$, \textbf{then} let $p \gets pq_e$\label{line-end}\;
}
let $F \gets 1$ with probability $\frac{w_u(H)}{4pp_0}$, where $H = (V_u, \+E)$\label{line-filter}\; 
$T \gets T - 1$\;
}
\textbf{if} $F = 1$ \text{then} \Return $H=(V_u,\+E)$; \textbf{else} \Return $\perp$\label{line-sample-return}\label{line-last-sample}\;
\end{algorithm}

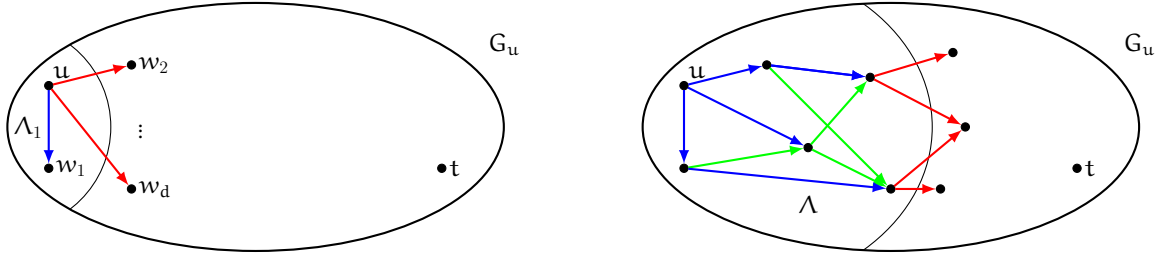
\begin{figure}[htbp]
  \centering
  \begin{subfigure}[t]{0.45\textwidth}
    \centering
    \begin{tikzpicture}[scale=0.55, inner sep=2pt, transform shape]
      \draw (0,0) ellipse (6 and 3) [thick];
      \draw (6,2) node {\LARGE $G_{u}$};
      \draw (-5,1) node [draw,fill,shape=circle,color=black, label=70:{\LARGE $u$}] (v) {};
      \draw (4.5,-1) node [draw,fill,shape=circle,color=black, label=0:{\LARGE $t$}] {}; 
      \draw (-5,0) node [label=180:{\LARGE $\Lambda_1$}] (Lambda) {};      
      \begin{scope}
        \clip (0,0) ellipse (6 and 3);
        \draw (-6,0) ellipse (2.5 and 2.5);
      \end{scope}
      
      \draw (-5,-1) node [draw,fill,shape=circle,color=black, label=-0:{\LARGE $w_1$}] (u1) {};
      \draw (-3,1.5) node [draw,fill,shape=circle,color=black, label=0:{\LARGE $w_2$}] (u2) {};
      \draw (-3,-0) node [label=0:{\LARGE $\vdots$}] (udots) {};
      \draw (-3,-1.5) node [draw,fill,shape=circle,color=black, label=0:{\LARGE $w_d$}] (u3) {};       
       \path (v) edge [->, color=blue, thick] (u1);
       \path (v) edge [->, color=red, thick] (u2);
       \path (v) edge [->, color=red, thick] (u3);
    \end{tikzpicture}
    \caption{At the start, we consider the first edge $(u,w_1)$. 
    To compute its marginal, we estimate two reliability, where the source is either $\Lambda_1=\{u,w_1\}$ (shown in picture) or just $u$, respectively.}
  \end{subfigure}
  \hspace{0.05\textwidth}
  \begin{subfigure}[t]{0.45\textwidth}
    \centering
    \begin{tikzpicture}[scale=0.55, inner sep=2pt, transform shape]
      \draw (0,0) ellipse (6 and 3) [thick];
      \draw (6,2) node {\LARGE $G_{u}$};
      \draw (-5,1) node [draw,fill,shape=circle,color=black, label=70:{\LARGE $u$}] (v) {};
      \draw (4.5,-1) node [draw,fill,shape=circle,color=black, label=0:{\LARGE $t$}] {}; 
      \draw (-2,-1.5) node [label=270:{\LARGE $\Lambda$}] (Lambda) {};      
      \begin{scope}[] 
        \clip (0,0) ellipse (6 and 3);
        \draw (-4,0) ellipse (5 and 4);
      \end{scope}
      
      \draw (-5,-1) node [draw,fill,shape=circle,color=black] (u1) {};
      \draw (-3,1.5) node [draw,fill,shape=circle,color=black] (u2) {};
      \draw (-2,-0.5) node [draw,fill,shape=circle,color=black] (u3) {};
      \draw (-0.5,1.2) node [draw,fill,shape=circle,color=black] (u4) {};
      \draw (0,-1.5) node [draw,fill,shape=circle,color=black] (u5) {};
       \path (v) edge [->, color=blue, thick] (u1);
       \path (v) edge [->, color=blue, thick] (u2);
       \path (v) edge [->, color=blue, thick] (u3);
       \path (u1) edge [->, color=green, thick] (u3);       
       \path (u2) edge [->, color=blue, thick] (u4);
       \path (u2) edge [->, color=green, thick] (u5);       
       \path (u2) edge [->, color=blue, thick] (u4);
       \path (u3) edge [->, color=green, thick] (u4);
       \path (u3) edge [->, color=green, thick] (u5);
       \path (u1) edge [->, color=blue, thick] (u5);

       \draw (1.5,1.8) node [draw,fill,shape=circle,color=black] (inv1) {};
       \draw (1.8,0) node [draw,fill,shape=circle,color=black] (inv2) {};
       \draw (1.2,-1.5) node [draw,fill,shape=circle,color=black] (inv3) {};
       \path (u4) edge [->, color=red, thick] (inv1);
       \path (u4) edge [->, color=red, thick] (inv2);
       \path (u5) edge [->, color=red, thick] (inv2);
       \path (u5) edge [->, color=red, thick] (inv3);
     \end{tikzpicture}
     \caption{As \Cref{alg-sample} progresses, there are chosen edges $\+E$ (blue),
     not chosen edges $E_1\setminus\+E$ (green), and the boundary edges (red).
   The set $\Lambda$ contains vertices reachable from $u$ using only $\+E$.}
  \end{subfigure}
  \caption{An illustration of sampling from $\pi_{u}$}\label{fig:sample-example}
\end{figure}

\begin{remark}[Crash of \textsf{Sample}]
The subroutine \textsf{Sample} (\Cref{alg-sample}) may crash in following cases: (1) in \Cref{line-pick}, $\partial \Lambda = \emptyset$; (2) in \Cref{line-sample-c}, $q_e c_0 + (1-q_e)c_1 = 0$; (3) in \Cref{line-filter}, $\frac{w_u(H)}{4p_0p} > 1$; (4) in \Cref{line-sample-return}, $F = 0$.
If it crashes, we stop \Cref{alg-main} immediately and output $\widetilde{R}_s = 0$.
\end{remark}


The algorithm needs $c_0$ and $c_1$ in order to compute the marginal probability of $e$ in \Cref{line-sample-c}.
The quantity~$c_0$ is an estimate to the reliability conditional on $e$ not selected and all choices made so far (namely $E_H \cap E_1 = \+E$).
Similarly, $c_1$ is an estimate to the reliability conditional on $e$ selected and $E_H \cap E_1 = \+E$.
Thus, if $\textsf{ApproxCount}$ returns exact values of $c_0$ and $c_1$, then
\begin{align*}
	\Pr_{H=(V_u,E_H)\sim \pi_u}\left[ e \in \+E_H \mid E_H \cap E_1 = \+E \right] = \frac{(1-q_e)c_1}{q_ec_0 + (1-q_e)c_1}.
\end{align*}
However, $\textsf{ApproxCount}$ can only approximate the reliabilities $c_0$ and $c_1$.
To handle the error from \textsf{ApproxCount}, our algorithm maintains a number $p$,
which is the probability of selecting the edges in $\+E$ and not selecting those in $E_1\setminus \+E$.
By the time we reach \Cref{line-filter}, $p$ becomes the probability that $H$ is generated by the algorithm. 
Then the algorithm uses a filter (with filter probability $\frac{w_u(H)}{4p_0p}$) to correct the distribution of $H$.
Conditional on passing the filter, $H$ is a perfect sample from $\pi_u$.
The detailed analysis of the error is given in \Cref{lemma-sample} and in \Cref{sec-main-analysis}.

Before we go to the $\textsf{ApproxCount}$ algorithm, we state one important property of \Cref{alg-sample}.
The topological ordering in \Cref{line-pick} is the ordering $s = v_1, v_2,\ldots,v_n = t$ in $G$. For any two vertices $v,v'$, we write $v \prec v'$ if $v = v_i$, $v' = v_j$ and $i < j$. 

\begin{fact}\label{fact-topo}
For any path $u_1,u_2,\ldots,u_\ell$ in $G$, it holds that $u_1 \prec u_2 \prec \ldots\prec u_\ell$.	
\end{fact}

\begin{lemma}\label{lemma-loop}
  In \Cref{alg-sample}, the following property holds: at the beginning of every while-loop, for any $w \in \partial \Lambda$, $E_1 \cap E_w = \emptyset$, where $E_w$ is the edge set of the graph $G_w = (V_w,E_w)$ defined in~\eqref{eq-def-G_u}.
\end{lemma}
\begin{proof}
  For any $i$, we use we use $X^{(i)}$ to denote some (vertex or edge) set $X$ at the beginning of the $i$-th loop, where $X$ can be $\Lambda,\partial\Lambda,\+E,E_1,E_2$.
  \def\kk{k}
We prove the lemma by contradiction. Suppose at the beginning of $\kk$-th loop, there exists $w \in \partial \Lambda^{(\kk)}$ such that $E_1^{(\kk)} \cap E_w \neq \emptyset$. We pick an arbitrary edge $(v,v') \in E_1^{(\kk)} \cap E_w$. Since $(v,v') \in E_1^{(\kk)}$, there must exist $j < \kk$ such that $(v,v') \notin E_1^{(j)}$ but $(v,v') \in E_1^{(j+1)}$, which means that the algorithm picks the edge $(v,v')$ in the $j$-th loop. We will prove that such $j$ cannot exist, which is a contradiction.

Since $w \in \partial \Lambda^{(\kk)}$, there must exist $w' \in \Lambda^{(\kk)}$ such that $(w',w) \in E_2^{(\kk)} \subseteq E$, where $E$ is the set of edges in the input graph $G$. By the definition of $\Lambda^{(\kk)}$, there exists a path $w_0,w_1,\ldots,w_{t-1}$ such that
\begin{itemize}
	\item $w_0 = u$ and $w_{t-1} = w'$;
	\item for all $1\leq i \leq t- 1$, $(w_{i-1},w_{i}) \in \+E^{(\kk)}$.
\end{itemize}
Hence, $\{w_0,w_1,\ldots,w_{t-1}, w_t\}$, where $w_t = w$, is a path in $G$. 
Note that $(v,v')$ is an edge in the graph $G_w$. By the definition of $G_w$, $w\prec v'$ (the vertex $v$ may be $w$). 
By~\Cref{fact-topo}, we have
\begin{align}\label{eq-order}
w_0 \prec w_1 \prec \ldots \prec w_{t} \prec v'.
\end{align}
Note that $w_0 = u \in \Lambda^{(j)}$ for all $j \geq 1$.
Also, since $w_t = w \notin \Lambda^{(\kk)}$, $w_t \notin \Lambda^{(j)}$ for all $j < \kk$. 
Hence, for any $1\le j < \kk$, there is an index $i^* \in \{1,2\ldots,t\}$ such that $w_{i^*-1} \in \Lambda^{(j)}$ and $w_{i^*} \notin \Lambda^{(j)}$.
We claim that $(w_{i^*-1},w_{i^*}) \in E_2^{(j)}$. 
\begin{itemize}
	\item If $i^* = t$, then $(w_{t-1},w_t) \in E_2^{(\kk)}$. Since $E_2^{(\kk)} \subseteq E_2^{(j)}$,  $(w_{t-1},w_t) \in E_2^{(j)}$;
    \item If $i^* \leq t-1$, then $(w_{i^*-1},w_{i^*}) \in \+E^{(\kk)}$. Suppose $(w_{i^*-1},w_{i^*}) \notin E^{(j)}_2$. Then  $(w_{i^*-1},w_{i^*})$ has been scanned before the $j$-th loop, namely $(w_{i^*-1},w_{i^*}) \in E^{(j)}_1$. However, $w_{i^*} \notin \Lambda^{(j)}$, namely that $w_{i^*}$ cannot be reached using $\+E^{(j)}$. It implies that $(w_{i^*-1},w_{i^*}) \notin \+E^{(j)}$ as otherwise $w_{i^*}$ can be reached because $w_{i^*-1}\in \Lambda^{(j)}$. This means that $(w_{i^*-1},w_{i^*}) $ has been scanned, but not added into $\+E$, which implies that $(w_{i^*-1},w_{i^*})  \notin \+E^{(\kk)}$. A contradiction. Hence, $(w_{i^*-1},w_{i^*}) \in E^{(j)}_2$.
\end{itemize} 
Thus the claim holds.
It implies that $w_{i^*} \in \partial \Lambda^{(j)}$.
On the other hand, by~\eqref{eq-order}, $w_{i^*} \prec v'$.
The way the next edge is chosen in \Cref{line-pick} implies that
in the $j$-th loop, the algorithm may only choose an edge $(s,s')$ such that $s' \in \partial \Lambda^{(j)}$ and $s'\preccurlyeq w_{i^*}$. 
In particular, the vertex $s'$ cannot be $v'$ and the edge $(v,v')$ cannot be chosen.
As this argument holds for all $j < \kk$, it proves the lemma.
\end{proof}

\begin{corollary}\label{cor-loop}
 In \Cref{alg-sample}, in every while-loop,  $\Lambda_1 = \Lambda \cup \{w^*\}$.	
\end{corollary}
\begin{proof}
  Clearly $\Lambda \cup \{w^*\} \subseteq \Lambda_1$. Suppose there exists $w_0 \neq w^*$ such that $w_0 \in \Lambda_1\setminus \Lambda$. Then $w_0 \in V_u$ can be reached from $w^*$ through edges in $\+E$. The vertex $w_0$ must be in $G_{w^*}$, which implies $\+E \cap E_{w^*} \neq \emptyset$, and thus $E_1 \cap E_{w^*} \neq \emptyset$. A contradiction to \Cref{lemma-loop}.
\end{proof}

\subsection{Approximate counting}\label{sec:AC}
Our \textsf{ApproxCount} subroutine is used in both~\Cref{alg-main} and~\Cref{alg-sample}. 
To suit~\Cref{alg-sample}, $\textsf{ApproxCount}$ takes input $(V,E,\Lambda, (\widetilde{R}_w,S_w)_{w \in \partial \Lambda})$ 
and approximates the $\Lambda-t$ reliability $R_\Lambda \defeq \Pr_{\+G}\left[{\rch{\Lambda}{t}{\+G}}\right]$,
where $\+G$ is a random spanning subgraph of $G$ obtained by removing each edge $e$ independently with probability $q_e$, and $\rch{\Lambda}{t}{\+G}$ denotes the event that $\exists w \in \Lambda$ s.t.~$\rch{w}{t}{\+G}$. 
When called by \Cref{alg-main}, $\Lambda$ is a single vertex, 
and when called by \Cref{alg-sample}, $\Lambda$ is the set defined in \eqref{eqn:Lambda-def}.
Moreover, the input satisfies:
\begin{itemize}
	\item $G=(V,E)$ is a DAG containing the sink $t$ and each edge has a failure probability $q_e$ (for simplicity, we do not write $t$ and all $q_e$ explicitly in the input as they do not change throughout~\Cref{alg-main} and~\Cref{alg-sample}); 
	\item $\Lambda \subseteq V$ is a subset of vertices that act as sources and $\partial \Lambda \defeq \{w \in V \setminus \Lambda \mid \exists w' \in \Lambda \text{ s.t. } (w',w)\in E\}$;
    \item for any $w \in \partial\Lambda$, $\widetilde{R}_w$ is an approximation of the $w-t$ reliability in $G_w$ and $S_w$ is a set of $\ell$ approximate random samples from the distribution $\pi_w$,
      where $G_w$ is defined in \eqref{eq-def-G_u}.
\end{itemize}
All points above hold when \Cref{alg-count} is evoked by \Cref{alg-main} or \Cref{alg-sample}. The first two points are easy to verify and the last one is verified in \Cref{sec-analysis}.


The algorithm first rules out the following two trivial cases:
\begin{itemize}
	\item if $t \in \Lambda$, the algorithm returns 1;
	\item if $\Lambda$ cannot reach $t$ in graph $G$, the algorithm returns 0.\footnote{Since all $q_e < 1$, $R_\Lambda > 0$ if and only if $\rch{\Lambda}{t}{G}$.}
\end{itemize}

As we are dealing with the more general set-to-vertex reliability, we need some more definitions. Define
\begin{align}\label{eq-def-Omega_Lambda}
	\Omega_\Lambda \defeq \{ H=(V,E_H) \mid E_H \subseteq E \land \rch{\Lambda}{t}{H} \}.
\end{align}
For any subgraph $H = (V, E_H)$ of $G = (V,E)$, define the weight function
\begin{align}\label{eqn:w_Lambda}
	w_\Lambda(H) \defeq \begin{cases}
		\prod_{e \in E_H}(1-q_e)\prod_{f \in E \setminus E_H}q_f &\text{if } \rch{\Lambda}{t}{H};\\
		0 &\text{if } \nrch{\Lambda}{t}{H}.
	\end{cases}
\end{align} 
Define the distribution $\pi_\Lambda$, whose support is $\Omega_{\Lambda}$,  by
\begin{align}\label{eqn:pi_Lambda}
	\pi_\Lambda(H) \defeq \frac{w_\Lambda(H)}{R_\Lambda}, \quad\text{where } R_\Lambda \defeq \sum_{H \in \Omega_\Lambda}w_\Lambda(H).
\end{align}

Let $\partial\Lambda$ be listed as $\{u_1,\dots,u_d\}$ for some $d\in[n]$.
Note that $d \geq 1$ because $R_{\Lambda} > 0$ and $t \notin \Lambda$.
To estimate $R_\Lambda$, we first write $\Omega_\Lambda$ in~\eqref{eq-def-Omega_Lambda} as a union of $d$ sets. 
For each $1 \leq i \leq d$, define
\begin{align*}
\Omega_\Lambda^{(i)} \defeq \left\{ H=(V,E_H) \mid (E_H \subseteq E) \,\land\, (\exists u \in \Lambda, \text{ s.t. } (u,u_i)\in E ) \, \land \, (\rch{u_i}{t}{H}) \right\}.
\end{align*}

\begin{lemma}\label{lemma-union}
If $t \notin \Lambda$, $\Omega_\Lambda = \cup_{i=1}^d \Omega_\Lambda^{(i)}$.
\end{lemma}
\begin{proof}
We first show $\cup_{i=1}^d \Omega_\Lambda^{(i)} \subseteq \Omega_\Lambda$. Fix any $H \in \cup_{i=1}^d \Omega_\Lambda^{(i)}$, say $H \in \Omega_\Lambda^{(i^*)}$ ($i^* \in [d]$ may not be unique, in which case we pick an arbitrary one). Then $\Lambda$ can reach $t$ in $H$, because we can first move from $\Lambda$ to $u_{i^*}$ and then move from $u_{i^*}$ to $t$. This implies $H \in \Omega_\Lambda$.
We next show $ \Omega_\Lambda \subseteq \cup_{i=1}^d \Omega_\Lambda^{(i)}$. Fix any $H \in \Omega_\Lambda$. There is a path  from $\Lambda$ to $t$ in $H$. Say the path is $w_1, w_2,\ldots, w_p = t$. Then $w_1 \in \Lambda$ (hence $w_1 \neq t$) and $w_2 = u_{i^*}$ for some $i^* \in [d]$.  Hence, $H$ contains the edge $(w_1,u_{i^*})$ and $\rch{u_{i^*}}{t}{H}$. This implies $H \in \cup_{i=1}^d \Omega_\Lambda^{(i)}$.
\end{proof}

Similar to \eqref{eqn:pi_Lambda}, we define $\pi^{(i)}$,\footnote{One may notice that if $\Omega^{(i)}_\Lambda = \emptyset$, then $\pi^{(i)}_\Lambda$ is not well-defined. \Cref{remark-welldfine} explains that $\Omega^{(i)}_\Lambda = \emptyset$ never happens because of \Cref{lem-input-condition}.} whose support is $\Omega_{\Lambda}^{(i)}$, by
\begin{align}
  \pi^{(i)}_\Lambda(H) &\defeq \frac{w_\Lambda(H)}{R^{(i)}_\Lambda }, \quad\text{where } R^{(i)}_\Lambda \defeq \sum_{H \in \Omega^{(i)}_{\Lambda}}w_\Lambda(H).\label{eqn:pi^i_Lambda}
\end{align}
In order to perform the Karp-Luby style estimation, we need to be able to do the following three things:
\begin{enumerate}
  \item compute the value $R_\Lambda^{(i)}$ for each $i \in [d]$;
  \item draw samples from $\pi_\Lambda^{(i)}$ for each $i \in [d]$;
  \item given any $i \in [d]$ and $H \in \Omega_\Lambda$, determine whether $H \in \Omega_\Lambda^{(i)}$.
\end{enumerate}
Suppose we can do them for now.\footnote{We will do (1) and (2) approximately rather than exactly. This will incur some error that will be controlled later.}
Consider the following estimator $Z_\Lambda$:
\begin{enumerate}
  \item draw an index $i\in [d]$ such that $i$ is drawn with probability proportional to $R_\Lambda^{(i)}$;
  \item draw an sample $H$ from $\pi_\Lambda^{(i)}$;
  \item let $Z_\Lambda \in \{0,1\}$ indicate whether $i$ is the smallest index $j \in [d]$ satisfying $H \in \Omega_\Lambda^{(j)}$. 
\end{enumerate}
It is straightforward to see that 
\begin{align}
	\Ex[Z_\Lambda] &= \sum_{i=1}^d \frac{R^{(i)}_\Lambda}{\sum_{j=1}^d R^{(j)}_\Lambda } \sum_{H \in \Omega^{(i)}_\Lambda} \pi^{(i)}_\Lambda(H) \cdot \boldsymbol{1}\left[ H \in \Omega_\Lambda^{(i)} \land \tp{ \forall j < i, H \notin \Omega_\Lambda^{(j)}}\right]\notag\\
	&= \frac{\sum_{i=1}^d \sum_{H \in \Omega_\Lambda^{(i)}} w_\Lambda(H) \cdot \boldsymbol{1} \left[H \in \Omega_\Lambda^{(i)} \land \tp{\forall j < i, H \notin \Omega_\Lambda^{(j)}}\right] }{\sum_{j=1}^d R^{(j)}_\Lambda }\notag\\
    \text{(by \Cref{lemma-union})}\quad &= \frac{ R_\Lambda }{\sum_{j=1}^d R^{(j)}_\Lambda} \geq \frac{1}{d} \geq \frac{1}{n},\label{eqn:expectation-lb}
\end{align}
where the first inequality holds because each $H$ belongs to at most $d$ different sets.
Since $Z_\Lambda$ is a $0$/$1$ random variable, $\Var{}{Z_\Lambda} \leq 1$. 
Hence, we can first estimate the expectation of $Z_\Lambda$ by repeating the process above and taking the average, and then use $\Ex[Z_\Lambda]\sum_{j=1}^d R^{(j)}_\Lambda$ as the estimator $\widetilde{R}_\Lambda$ for $R_\Lambda$. 
However, in the input of \textsf{ApproxCount}, we only have estimates $\widetilde{R}_{u_i}$ and a limited set of samples $S_{u_i}$ for each $u_i \in \partial\Lambda$.
Our algorithm will need to handle these imperfections.

The third step of implementing the estimator $Z_\Lambda$ is straightforward by a BFS.
For the first step, $R_\Lambda^{(i)}$ can be easily computed given $R_{u_i}$, and we will use the estimates $\widetilde{R}_{u_i}$ instead.
For the second step, define
\begin{align*}
	\delta_\Lambda(u_i) \defeq \{ (w,u_i) \in E \mid w \in \Lambda \}.
\end{align*}
To sample from $\pi^{(i)}_\Lambda$, we shall sample at least one edge in $\delta_\Lambda(u_i)$,
sample a subgraph $H$ from $\pi_{u_i}$, and add all other edges independently.
These are summarized by the next lemma.
\begin{lemma}\label{lemma-expand-R}
  For any $i \in [d]$, it holds that $R^{(i)}_\Lambda = \tp{1-\prod_{u \in \delta_\Lambda(u_i)} q_{(u,u_i)}}R_{u_i}$ and a random sample $H'=(V, E_{H'}) \sim \pi^{(i)}_\Lambda$ can be generated by the following procedure:
\begin{itemize}
	\item sample $H=(V_{u_i},E_H) \sim \pi_{u_i}$;
	\item let $E_{H'} = E_H \cup D$, where $D \subseteq \delta_\Lambda(u_i)$ is a random subset with probability proportional to 
      \begin{align}\label{eqn:D-lemma-4}
        \boldsymbol{1}[D \neq \emptyset] \cdot \prod_{e \in \delta_\Lambda(u_i) \cap D}(1-q_e) \prod_{f \in \delta_\Lambda(u_i) \setminus D}q_f;
      \end{align}
	\item add each $e \in E \setminus (E_{u_i} \cup \delta_\Lambda(u_i))$ into $E_{H'}$ independently with probability $1-q_e$.
\end{itemize}
All three steps above handle mutually exclusive edge sets and thus are mutually independent.
\end{lemma}
\begin{proof}
If each edge $e$ in $G$ is removed independently with probability $q_e$, we have a random spanning subgraph $\+G=(V,\+E)$. By the definition of $R^{(i)}_\Lambda$, 
\begin{align*}
	R^{(i)}_\Lambda &= \Pr\left[\exists u \in \Lambda, \text{ s.t. } (u,u_i)\in \+E  \,\land\, \rch{u_i}{t}{\+G}\right]\\
	 &= \Pr\left[\exists u \in \Lambda, \text{ s.t. } (u,u_i)\in \+E \right] \cdot \Pr\left[\rch{u_i}{t}{\+G} \mid \exists u \in \Lambda, \text{ s.t. } (u,u_i)\in \+E  \right].
\end{align*}
It is easy to see $\Pr[\exists u \in \Lambda, \text{ s.t. } (u,u_i)\in \+E ] = 1-\prod_{u \in \delta_\Lambda(u_i)} q_{(u,u_i)}$. For the second conditional probability, note that the event $\rch{u_i}{t}{\+G}$ depends only on the randomness of edges in graph $G_{u_i}$. In other words, for any edges $e \in E \setminus E_{u_i}$, whether or not $e$ is removed has no effect on the $u_i-t$ reachability. Due to acyclicity, all edges in $\delta_\Lambda(u_i)$ are not in the graph $G_{u_i}$. We have
\begin{align*}
	R^{(i)}_\Lambda = \tp{1-\prod_{u \in \delta_\Lambda(u_i)} q_{(u,u_i)}}R_{u_i}.
\end{align*}

By definition \eqref{eqn:pi^i_Lambda}, 
$\pi^{(i)}_\Lambda$ is the distribution of $\+G=(V,\+E)$ conditional on $\exists u \in \Lambda, \text{ s.t.~} (u,u_i)\in \+E$ and $\rch{u_i}{t}{\+G}$. For any graph $H' = (V,E_{H'})$ satisfying $\exists u \in \Lambda, \text{ s.t. } (u,u_i)\in E_{H'}$ and $\rch{u_i}{t}{{H'}}$, we have
\begin{align*}
\pi^{(i)}_\Lambda(H') &= \Pr\left[\+G = H' \mid \exists u \in \Lambda, \text{ s.t. } (u,u_i)\in \+E \,\land\, \rch{u_i}{t}{\+G}\right]\\
	 &= \frac{\Pr[\+G = H']}{R^{(i)}_\Lambda}=\frac{\Pr[\+G = H']}{(1-\prod_{u \in \delta_\Lambda(u_i)} q_{(u,u_i)})R_{u_i}}\\
&= \prod_{e \in E_{H'}: e \in E \setminus E_{u_i}}(1-q_e)\prod_{f \notin E_{H'}: f \in E \setminus E_{u_i}}q_f \cdot \frac{w_{u_i}(H'[V_{u_i}])}{(1-\prod_{u \in \delta_\Lambda(u_i)} q_{(u,u_i)})R_{u_i}}\\
&= \pi_{u_i}(H'[V_{u_i}]) \cdot \frac{ \prod_{e \in \delta_{\Lambda}(u_i) \cap E_{H'}}(1-q_e) \prod_{f \in \delta_\Lambda(u_i) \setminus E_{H'} }q_f }{1-\prod_{u \in \delta_\Lambda(u_i)} q_{(u,u_i)}}\\
 &\qquad\cdot \prod_{\substack{ e \in E_{H'}:\\ e \in E \setminus (E_{u_i} \cup \delta_{\Lambda}(u_i) ) } }(1-q_e)\prod_{ \substack{ f \notin E_{H'}:\\ f \in E \setminus (E_{u_i} \cup \delta_{\Lambda}(u_i) ) }}q_f.
\end{align*}
The probability above exactly matches the procedure in the lemma.
\end{proof}
The first sampling step in \Cref{lemma-expand-R} can be done by directly using the samples from $S_{u_i}$.
We still need to show that the second step in \Cref{lemma-expand-R} can be done efficiently.
\begin{lemma}\label{lem-eff}
There is an algorithm such that
given a set $S= \{1,2,\ldots,n\}$ and $n$ numbers $0 \leq q_1, q_2,\ldots,q_n < 1$, it return a random non-empty subset $D \subseteq S$ with probability proportional to $\boldsymbol{1}[D\neq \emptyset]\prod_{i \in D}(1-q_i)\prod_{j \in S \setminus D}q_j$ in time $O(n)$.
\end{lemma}
\begin{proof}
Note that $D$ can be obtained by sampling each $i$ in $S$ independently with probability $1-q_i$ conditional on the outcome is non-empty. A natural idea is to use rejection sampling, but $1 - \prod_{i=1}^n q_i$ can be very small. 
Here we do this in a more efficient way.

We view any subset $D \subseteq S$ as an $n$-dimensional vector $\sigma\in\{0,1\}^S$. We sample $\sigma_i$ for $i$ from $1$ to $n$ one by one. In every step, conditional on $\sigma_1=c_1,\sigma_2=c_2,\ldots,\sigma_{i-1} = c_{i-1} \in \{0,1\}$, we compute the marginal of $\sigma_i$ and sample from the marginal. The marginal can be computed as follows: for any $c_i \in \{0,1\}$,
\begin{align*}
	\Pr\left[\sigma_i = c_i \mid \forall j < i, \sigma_j=c_j\right]= \frac{\Pr[ \forall j \leq i, \sigma_j=c_j]}{\Pr[\forall j \leq i-1, \sigma_j=c_j]}.
\end{align*}
It suffices to compute $\Pr[ \forall j \leq i, \sigma_j=c_j]$ for any $1\leq i \leq n$. Let $\Omega$ denote the set of all assignments for $\{i+1,i+2,\ldots,n\}$. For any $\tau \in \Omega$, $\tau$ is an $(n-i)$-dimensional vector, where $\tau_k \in \{0,1\}$ is the value for $k \geq i + 1$.
We use $(c_j)_{j\leq i}+\tau$ to denote an $n$-dimensional vector.
For any $j$, let $f_j(0) = q_j$ and $f_j(1) = 1 - q_j$.
Note that
\begin{align*}
	\Pr\left[\forall j \leq i, \sigma_j=c_j\right] = \frac{\prod_{j=1}^i f_j(c_j) \sum_{\tau \in \Omega}\prod_{k = i+1}^n f_k(\tau_k) \boldsymbol{1}\left[ (c_j)_{j\leq i}+\tau \text{ is not zero vector} \right]}{1 - \prod_{j=1}^nq_i}.
\end{align*}
Hence, if $c_1 + c_2 + \ldots + c_i \geq 1$, then
\begin{align*}
	\Pr\left[\forall j \leq i, \sigma_j=c_j\right] = \frac{\prod_{j=1}^i f_j(c_j)}{1 - \prod_{j=1}^nq_i}.
\end{align*}
If  $c_1 + c_2 + \ldots + c_i = 0$, then
\begin{align*}
\Pr\left[\forall j \leq i, \sigma_j=c_j\right] 	= \frac{(1 - \prod_{k=i+1}^nq_k)\prod_{j=1}^i f_j(c_j)}{1 - \prod_{j=1}^nq_i}.
\end{align*}
Hence, every conditional marginal can be computed by the formula above in time $O(n)$. An naive sampling implementation takes $O(n^2)$ time to compute all the marginal probabilities, but it is not hard to see that a lot of prefix or suffix products can be reused and the total running time of sampling can be reduced to~$O(n)$.
\end{proof}

Now, we are almost ready to describe \textsf{ApproxCount} (\Cref{alg-count}). For any $u_i$, we have an approximate value $\widetilde{R}_{u_i}$ of ${R}_{u_i}$ and we also have a set $S_{u_i}$ of $\ell$ approximate samples from the distribution $\pi_{u_i}$. By \Cref{lemma-expand-R} and \Cref{lem-eff}, we can efficiently approximate $R^{(i)}_\Lambda$ and generate approximate samples from $\pi^{(i)}_\Lambda$. 
Hence, we can simulate the process described below \Cref{lemma-union} to estimate $R_\Lambda$.

However, to save the number of samples, there is a further complication.
Our algorithm estimates the expectation of $\Ex[Z_{\Lambda}]$ in two rounds and then takes the median of estimators.
Recall \eqref{eq-def-ell}.
We further divide~$\ell$ by introducing the following parameters:
\begin{align*}
  \ell = B\ell_0,\quad B \defeq 60n + 150m,\quad \ell_0 \defeq \ell_1 + 500\ell_2, \quad \ell_1 \defeq 400n, \quad \ell_2 \defeq\lceil 10^4 n^2\max\{m^2,\epsilon^{-2}\}\rceil.
\end{align*} 
Then we do the following.
\begin{itemize}
  \item For any $u_i \in \partial \Lambda$, we divide all $\ell$ samples in $S_{u_i}$ into $B$ blocks, each containing $\ell_0$ samples. Denote the $B$ blocks by $S_{u_i}^{(1)},S_{u_i}^{(2)},\ldots,S_{u_i}^{(B)}$.
    Each block here is used for one estimatior.
  \item For each $i \in [d]$ and $j \in [B]$, we further partition $S_{u_i}^{(j)}$ into two multi-sets $S_{u_i}^{(j,1)}$ and $S_{u_i}^{(j,2)}$, 
    where $S_{u_i}^{(j,1)}$ has $\ell_1$ samples and $S_{u_i}^{(j,2)}$ has $500\ell_2 $ samples.
    These two sets are used for the two rounds, respectively.
  \item For each $j \in [B]$, we do the following two round estimation:
	\begin{enumerate}
		\item use samples in $(S_{u_i}^{(j,1)})_{i \in [d]}$ to obtain a \emph{constant-error} estimation $\widehat{Z}_{\Lambda}^{(j)}$ of $\Ex[Z_{\Lambda}]$;
		\item use $\widehat{Z}_{\Lambda}^{(j)}$ and samples in $(S_{u_i}^{(j,2)})_{i \in [d]}$ to obtain a more accurate  \emph{estimation} $\widetilde{Z}_{\Lambda}^{(j)}$ of $\Ex[Z_{\Lambda}]$;
        \item let $Q^{(j)}_\Lambda \gets \widetilde{Z}_\Lambda^{(j)}\sum_{i=1}^d \widetilde{R}_{\Lambda}^{(i)}$, where $\widetilde{R}_{\Lambda}^{(i)}\defeq\big(1-\prod_{u \in \delta_\Lambda(u_i)} q_{(u,u_i)}\big) \widetilde{R}_{u_i}$ for each $i\in [d]$.
	\end{enumerate}
	\item Return the median number $\widetilde{R}_{\Lambda} \defeq \text{median}\left\{Q_{\Lambda}^{(1)},Q_{\Lambda}^{(2)},\ldots,Q_{\Lambda}^{(B)}\right\}$.
\end{itemize}

A detailed description of \textsf{ApproxCount} is given in \Cref{alg-count}.
It uses a subroutine $\textsf{Estimate}$, which generates the Karp-Luby style estimator and is described in \Cref{alg-es}.

\begin{algorithm}[h] 
  \caption{$\textsf{ApproxCount}\tp{V,E,\Lambda,(\widetilde{R}_{w}, S_w)_{w \in \partial \Lambda}}$}\label{alg-count} 
  \KwIn{a graph $G=(V,E)$, a subset $\Lambda \subseteq V$, all $(\widetilde{R}_{w}, S_w)$ for $w \in \partial \Lambda$, where $\partial \Lambda = \{w \in V \setminus \Lambda \mid \exists w' \in \Lambda \text{ s.t. } (w',w)\in E\}$;}
  \KwOut{an estimator  $\widetilde{R}_\Lambda$ of $R_\Lambda$}
\textbf{if} $t \in \Lambda$, \textbf{then} \Return $0$; \textbf{if} $\Lambda$ cannot reach $t$ in $G$, \textbf{then} \Return 1\label{line-trival}\;
\For{$u_i \in \partial \Lambda$}{
   $\widetilde{R}_{\Lambda}^{(i)} \gets \big(1-\prod_{u \in \delta_\Lambda(u_i)} q_{(u,u_i)}\big) \widetilde{R}_{u_i}$\label{line-R-tilde}\;
   partition  $S_{u_i}$ (arbitrarily) into $B$ multi-sets, denoted by $S_{u_i}^{(j)}$ for $j\in [B]$, where $B = 60n + 150m$ and each $S_{u_i}^{(j)}$ has $\ell_0 = \ell_1 + 500\ell_2 $ samples\; 
   for each $j\in[B]$, partition $S_{u_i}^{(j)}$ further into two multi-sets $S_{u_i}^{(j,1)}$ and $S_{u_i}^{(j,2)}$, where $|S_{u_i}^{(j,1)}| = \ell_1$ and $|S_{u_i}^{(j,2)}|=500\ell_2$\; 
}
\For{$j$ from $1$ to $B$}{
$\widehat{Z}_{\Lambda}^{(j)} \gets \textsf{Estimate}\tp{(S_{u_i}^{(j,1)})_{i \in  [d]}, \ell_1, \ell_1} $\label{line-round-1}\;
$\widetilde{Z}^{(j)}_\Lambda \gets \textsf{Estimate}\tp{(S_{u_i}^{(j,2)})_{i \in [d]}, 500\ell_2, 25\ell_2\cdot \min\{ 2/ \widehat{Z}_{\Lambda}^{(j)},4n\} }$\label{line-round-2}\;
$Q^{(j)}_\Lambda \gets \widetilde{Z}^{(j)}_\Lambda  \sum_{i=1}^d \widetilde{R}_{\Lambda}^{(i)}$\;
}
\Return $\widetilde{R}_{\Lambda} \defeq \text{median}\left\{Q_{\Lambda}^{(1)},Q_{\Lambda}^{(2)},\ldots,Q_{\Lambda}^{(B)} \right\}$\label{line-Q}\;
\end{algorithm}

\begin{algorithm}[h] 
  \caption{$\textsf{Estimate}\tp{(S^{\-{es}}_{u_i})_{i \in [d]}, \ell_{\-{es}} ,T}$}\label{alg-es} 
  \KwIn{a set of samples $S^{\-{es}}_{u_i}$ for each $i \in [d]$, where $|S^{\-{es}}_{u_i}| = \ell_{\-{es}}$, a threshold $T$}
  \KwOut{an estimator  $Z_{\-{es}}$ of $\Ex[Z_\Lambda]$}
 for each $i \in [d]$, let $c_i = 0$\;
\For{$k$ from $1$ to $T$}{
	draw an index $i\in [d]$ such that $i$ is drawn with probability proportional to $\widetilde{R}_{\Lambda}^{(i)}$\; 
	$c_i \gets c_i + 1$\;
	\textbf{if } $c_i > \ell_{\-{es}}$ \textbf{ then } \Return $0$\label{line-break}\;
	let $H=(V_{u_i},E_H)$ be the $c_i$-th sample from $S^{(j)}_{u_i}$\label{line-pick-sample}\; 
	do the following transformation on $H$ to get $H'=(V,E_{H'})$\;
		 $\bullet\quad$  let $E_{H'} \gets E_{H}$\;
         $\bullet\quad$ draw $D \subseteq \delta_\Lambda(u_i)$ with probability proportional to \eqref{eqn:D-lemma-4}, and let $E_{H'} \gets E_{H'} \cup D$\;
	 $\bullet\quad$ for each $e \in E \setminus (E_{u_i} \cup \delta_\Lambda(u_i))$, add $e$ into $E_{H'}$ independently with probability $1-q_e$\label{line-H'}\;
	
	let $Z_{\-{es}}^{(k)} \in \{0,1\}$ indicate whether $i$ is the smallest index $t \in [d]$ satisfying $H' \in \Omega_\Lambda^{(t)}$\label{line-test}\;
	}
\Return $Z_{\-{es}} \defeq \frac{1}{T}\sum_{k = 1}^T Z^{(k)}_{\-{es}} $\; 
\end{algorithm}

Each time \Cref{alg-count} finishes, its input $(V,E,\Lambda)$ and output $\widetilde{R}_{\Lambda}$ are stored in the memory. If \Cref{alg-count} is ever evoked again with the same input parameters $(V,E,\Lambda)$, we simply return $\widetilde{R}_{\Lambda}$ from the memory.

\Cref{alg-count} first obtains the constant-error estimation $\widehat{Z}_\Lambda^{(j)}$ in~\Cref{line-round-1}.
Next, it puts $\widehat{Z}_\Lambda^{(j)} $ into the parameters and run the subroutine \textsf{Estimate} again to get a more accurate estimation $\widetilde{Z}_\Lambda^{(j)} $. 
The benefit of this two-round estimation is that we can save the number of samples maintained for each vertex. 
It costs only a small number of samples to get the crude estiamtion, 
but the crude estimation carries information of the ratio $\frac{\sum_{t\in[d]}R_{\Lambda}^{(t)}}{R_{\Lambda}}$,
which allows us to fine tune the number of samples required per vertex for the good estimation.
To be more specific, in the second call of the subroutine $\textsf{Estimate}$, the number of overall samples, namely the parameter $T$ which depends on $\widehat{Z}_\Lambda^{(j)}$,
can still be as large as $\Omega(\ell_2 n)$ in the worst case,
and yet each $S^{(j,2)}_{u_i}$ has only $O(\ell_2)$ samples in the block.
In the analysis (\Cref{lemma-ac}), we will show that this many samples per vertex suffice with high probability and \Cref{line-break} of \Cref{alg-es} (the failure case) is executed with low probability.
The reason is that, roughly speaking, large $T$ means large overlap among $\Omega_{\Lambda}^{(t)}$'s,
and the chance of hitting each vertex is roughly the same, resulting in the number of samples required per vertex close to the average.
Conversely, small $T$ means little overlap, and some vertex or vertices may be sampled much more often than other vertices,
but in this case the overall number of samples required, namely $T$, is small anyways.
To summarize, with this two-round procedure, $O(\ell)$ overall samples per vertex are enough to obtain an estimation with the desired accuracy and high probability.

\section{Analysis} \label{sec-analysis}
In this section we analyze all the algorithms.

\subsection{Analysis of \textsf{Sample}}
Let $G = (V,E)$ be the input graph of \Cref{alg-main}. Consider the subroutine $\textsf{Sample}\tp{v_k,(\widetilde{R}_{w}, S_w)_{w \in \{v_{k+1},v_{k+2},\ldots,v_n\}},\widetilde{R}_{v_k}}$ as being called by \Cref{alg-main}. Let $u\defeq v_k$, and the subroutine runs on the graph $G_u = (V_u,E_u)$.
In this section we consider a modified version of $\textsf{Sample}$, and handle the real version in \Cref{sec-main-analysis}.
Let $m$ denote the number of edges in $G$.
Suppose we can access an oracle $\+P$ satisfying:
\begin{itemize}
  \item given $u \in V_u$,  $\+P$ returns $p_0$ such that
	\begin{align}\label{eq-r1}
      1 - \frac{1}{10m}\leq \frac{p_0}{R(V_u,E_u,\{u\})} \leq 1 + \frac{1}{10m};
	\end{align}
  \item given any $E_2 \subseteq E_u$ and $\Lambda,\Lambda_1 \subseteq V$ in \Cref{line-c1} and \Cref{line-c2} of \Cref{alg-sample}, $\+P$ returns $c_0(V_u,E_2,\Lambda)$ and $c_1(V_u,E_2,\Lambda_1)$ such that
	\begin{align}
      1 - \frac{1}{10m} &\leq \frac{c_0(V_u,E_2,\Lambda)}{R(V_u,E_2,\Lambda)} \leq 1 + \frac{1}{10m}, \text{ and }\label{eq-r2}\\
      1 - \frac{1}{10m} &\leq \frac{c_1(V_u,E_2,\Lambda_1)}{R(V_u,E_2,\Lambda_1)} \leq 1 + \frac{1}{10m}.\label{eq-r3}
	\end{align}
\end{itemize}
Here, we use the convention $\frac{0}{0} = 1$ and $\frac{x}{0} = \infty$ for $x > 0$.
For any $V$, $E$ and $U \subseteq V_u$, $R(V,E,U)$ is the $U$-$t$ reliability in the graph $(V,E)$. The numbers $p_0,c_0$ and $c_1$ returned by $\+P$ can be random variables, but we assume that the inequalities above are always satisfied. Abstractly, one can view $\+P$ as a random vector $\+X_{\+P}$, where 
\begin{align*}
  \+X_{\+P} = \{p_0\}\cup\{c_0(V_u,E_2,\Lambda), c_1(V_u,E_2,\Lambda_1) \mid \text{ for all possible } V_u,E_2,\Lambda,\Lambda_1\}.
\end{align*}
The dimension of $\+X_{\+P}$ is huge because there may be exponentially many possible $V_u,E_2,\Lambda,\Lambda_1$ in \Cref{line-c1} and \Cref{line-c2}. 
The oracle first draw a sample $x_{\+P}$ of $\+X_{\+P}$, then answers queries by looking at $x_{\+P}$ on the corresponding coordinate. 
The conditions above are assumed to be satisfied with probability $1$. 
Note that this $\+X_{\+P}$ is only for analysis purposes and is not part of the real implementation.

The modified sampling algorithm replaces \Cref{line-p_0}, \Cref{line-c1} and \Cref{line-c2} of \Cref{alg-sample} by calling the oracle $\+P$. 
In that case we do not need the estimates $(\widetilde{R}_w,S_w)$ for $w=v_{k+1},\dots,v_n$ and $\widetilde{R}_{v_k}$,
and thus may assume that the input is only $u=v_k$.
Recall that $n$ is the number of vertices in the input graph.
\begin{lemma}\label{lemma-sample}
  Given any $u = v_k \in V$, the with probability at least $1 - (\eps / n)^{200}$,  the modified sampling algorithm does not crash and returns a perfect independent sample from the distribution $\pi_u$, where the probability is over the independent randomness $\+D_u$ inside the \textsf{Sample} subroutine. The running time is $\widetilde{O}(N(|E_u|+|V_u|))$, where $N$ is the time cost for one oracle call and $\widetilde{O}$ hides $\text{polylog}(n/\eps)$ factors. 
\end{lemma}

\begin{proof}
  Throughout this proof, we fix a sample $x_{\+P}$ of $\+X_{\+P}$ in advance. The oracle $\+P$ uses $x_{\+P}$ to answer the queries. We will prove that the lemma holds for any $x_{\+P}$ satisfying the three conditions above.

  We first describe an ideal sampling algorithm.
  The algorithm maintains the set $E_1,E_2$ and $\+E$ as in the \textsf{Sample} algorithm. At each step, we pick an edge $e$ according to~\Cref{line-pick} of \Cref{alg-sample}. We compute the conditional marginal probability of $\alpha_e = \Pr_{\+G=(V_u,E')\sim \pi_u} \left[ e \in E' \mid E' \cap E_1 = \+E \right]$, and add $e$ into $\+E$ with probability $\alpha_e$. Then we update $E_1, E_2$ and $\Lambda$. Once $t \in \Lambda$, $\alpha_e = 1-q_e$ for all $e \in E_2$ and we can add all subsequent edges independently.
  The ideal sampling algorithm returns an independent perfect sample.

The modified algorithm simulates the ideal process, but uses the oracle $\+P$ to compute each conditional marginal distribution $\alpha_e$. 
By the definition of conditional probability,
\begin{align}\label{eqn:alpha-e}
  \alpha_e = \frac{\Pr_{\+G=(V_u,E')\sim \pi_u} \left[ e \in E' \land  E' \cap E_1 = \+E \right]}
  {\Pr_{\+G=(V_u,E')\sim \pi_u} \left[ e \in E' \land E' \cap E_1 = \+E \right] + \Pr_{\+G=(V_u,E')\sim \pi_u} \left[ e \notin E' \land E' \cap E_1 = \+E \right]}.
\end{align}
Recall $E_2 = E_u \setminus E_1$. Let $E'_2 = E_2 \setminus e$. Recall $\Lambda_1$ is the set of vertices $u$ can reach if $\+E \cup \{e\}$ is selected. 
Conditional on that $\+E \cup \{e\}$ is selected, the probability that $u$ can reach $t$ is exactly the same as the probability~$\Lambda_1$ can reach $t$ in the remaining graph $(V_u, E'_2)$.
Then the numerator of \eqref{eqn:alpha-e} can be written as
\begin{align*}
	(1-q_e)\prod_{f \in E_1 \cap \+E}(1-q_f)\prod_{f' \in E_1 \cap \+E}q_{f'} \cdot R(V_u, E'_2, \Lambda_1),
\end{align*}
where $ R(V_u, E'_2, \Lambda_1)$ is the $\Lambda_1$-$t$ reliability in the graph $(V_u,E'_2)$.
Similarly, the second term of the denominator of \eqref{eqn:alpha-e} can be written as
\begin{align*}
	q_e\prod_{f \in E_1 \cap \+E}(1-q_f)\prod_{f' \in E_1 \cap \+E}q_{f'} \cdot R(V_u, E'_2, \Lambda).
\end{align*}
Putting them together implies 
\begin{align*}
	\alpha_e = \frac{(1-q_e)R(V_u, E'_2, \Lambda_1)}{(1-q_e)R(V_u, E'_2, \Lambda_1)+q_eR(V_u, E'_2, \Lambda)}.
\end{align*}
If $c_0$ and $c_1$ in \Cref{line-c1} and \Cref{line-c2} are exactly $R(V_u, E'_2, \Lambda)$ and $R(V_u, E'_2, \Lambda_1)$, then \Cref{line-strat} to \Cref{line-end} in \Cref{alg-sample} are the same as the ideal algorithm described above. Under this assumption, \Cref{alg-sample} cannot crash in \Cref{line-pick-sample} or \Cref{line-sample-c}.
Consider the modified algorithm in the lemma.
Note that the state of the algorithm can be uniquely determined by the pair $(E_2,\+E)$. 
By the assumption of $\+P$, we know that $R(V_u, E'_2, \Lambda) = 0$ if and only if $c_0 = 0$ and $R(V_u, E'_2, \Lambda_1) = 0$ if and only if $c_1 = 0$. 
Hence, any state $(E_2,\+E)$ appears in the modified algorithm with positive probability if and only if it appears in the ideal algorithm with positive probability. 
This implies the modified algorithm cannot crash in \Cref{line-pick-sample} or \Cref{line-sample-c}.

By the assumption of the oracle $\+P$ again, we have
\begin{align*}
  1- \frac{1}{4m} \leq \frac{10m-1}{10m+1} &\leq \frac{\frac{(1-q_e)c_1}{(1-q_e)c_1+q_e c_0}}{\alpha_e} \leq \frac{10m+1}{10m-1}\leq 1+ \frac{1}{4m},\\\
  1- \frac{1}{4m} \leq \frac{10m-1}{10m+1} &\leq \frac{\frac{q_e c_0}{(1-q_e)c_1+q_e c_0}}{1-\alpha_e} \leq \frac{10m+1}{10m-1} \leq 1+ \frac{1}{4m}.
\end{align*}
When the algorithm exits the whole loop, 
$u$ can reach $t$ and we have the remaining marginals exactly.

Finally, the algorithm gets a random subgraph $H$ and a value $p$, where $p=p(H)$ is the probability that the algorithm generates $H$. 
Note that there are at most $m$ edges in $E_u$.
Taking the product of all conditional marginals gives 
\begin{align*}
  \exp(-1/2) \leq \left(1- \frac{1}{4m}\right)^{m}  \leq \frac{p(H)}{\pi_u(H)} \leq \left(1+ \frac{1}{4m}\right)^{m} \leq \exp(1/4).
\end{align*}
Recall that 
\begin{align*}
	\pi_u(H)=\frac{w_u(H)}{R(V_u,E_u,u)}.
\end{align*}
By the assumption of the oracle $\+P$, we have 
\begin{align*}
  \frac{9}{10} \leq \frac{p_0}{R(V_u,E_u,u)} \leq \frac{11}{10}.
\end{align*}
The parameter $p_0 > 0$ because the input of \Cref{alg-sample} must satisfy $R(V_u,E_u,u) > 0$.
The filter probability $f = \Pr[F = 1 \mid H]$ in \Cref{line-filter} of \Cref{alg-sample} satisfies
\begin{align*}
	\frac{1}{16} \leq  f = \frac{w_u(H)}{4p(H)p_0} \leq 1.
\end{align*}
Hence, $f$ is a valid probability and $f \geq \frac{1}{16}$. The algorithm cannot crash in \Cref{line-filter}.
The algorithm outputs $H$ if $F = 1$.
By the analysis above, we know that $p(H) > 0\Leftrightarrow\pi_u(H) > 0$ and 
\begin{align*}
\Pr[\textsf{Sample} \text{ outputs } H] \propto p(H) \frac{w_u(H)}{p(H)p_0} = \frac{w_u(H)}{p_0} \propto w_u(H),
\end{align*}
where the last ``proportional to'' holds because $p_0$ is a constant (independent from $H$). Conditional on $F=1$, $H$ is a perfect sample. We repeat the process for $T = 1000\log \frac{n}{\eps}$ times, and each time the algorithm succeeds with probability at least $\frac{1}{16}$. The overall probability of success is at least $1 - (\eps/n)^{200}$.

The running time is dominated by the oracle calls. We can easily use data structures to maintain $\partial \Lambda$, $\Lambda$, $E_2$ and $\+E$. The total running time is $\tilde{O}((|V_u|+|E_u|)N)$.
\end{proof}

The above only deals with the modified algorithm.
Analysing the real algorithm relies on the analysis of \textsf{ApproxCount},
and we defer that to \Cref{sec-main-analysis}.

The \Cref{alg-sample} can only be evoked by \Cref{alg-main}. Fix $u = v_k$. Suppose we use \Cref{alg-sample} to draw samples from $\pi_u$.
For later analysis, we need to make clear how each random variable depends on various sources of randomness.
We abstract the modified algorithm as follows. The oracle $\+P$ is determined by a random vector $\+X_{\+P}$. The algorithm generates the inside independent randomness $\+D_u$. 
The algorithm constructs a random subgraph $H = H(\+X_{\+P},\+D_u)$ and a random indicator variable $F = F(\+X_{\+P},\+D_u)$, where $H$ and $F$ denote the random variables of the same name in the last line of \Cref{alg-sample}. \Cref{lemma-sample} shows that conditional on $F  = 1$, $H$ is an independent sample (independent from $\+X_{\+P}$) that follows $\pi_u$. We denote it by
\begin{align*}
	H(\+X_{\+P},\+D_u)|_{F(\+X_{\+P},\+D_u) = 1} \sim \pi_u.
\end{align*} 
Here, for any random variable $X$ and event $E$, we use $X|_E$ to denote the random variable $X$ conditional on $E$.
In fact, a following stronger result can be obtained from the above proof
\begin{align*}
	\forall x_{\+P} \in \Omega_{\+P}, \quad H|_{F(\+X_{\+P},\+D_u) = 1 \land \+X_{\+P} = x_{\+P}} \sim \pi_u.
\end{align*}
where  $\Omega_{\+P}$ denotes the support of $\+X_{\+P}$.
And it holds that 
\begin{align*}
	\forall x_{\+P} \in \Omega_{\+P}, \quad \Pr[F(\+X_{\+P},\+D_u) = 1 \mid \+X_{\+P} = x_{\+P}] \geq 1 - \frac{1}{(n/\eps)^{200}}.
\end{align*}
Note that the event $F(\+X_{\+P},\+D_u) = 1$ depends on the input random variable $\+X_{\+P}$. In the analysis in \Cref{sec-main-analysis}, we need to define a event $\+C$ such that $\Pr[\+C] \geq 1 -  (\eps/n)^{200}$, $\+C$ is independent from $\+X_{\+P}$ and $H(\+X_{\+P},\+D_u)|_{\+C} \sim \pi_u$. 
We actually define this event $\+C$ in a more refined probability space. The proof below defines this event explicitly. 
We also include an alternative, more conceptual, and perhaps simpler proof in \Cref{app-proof}, where the  event $\+C$ is defined implicitly. 

Consider the following algorithm $\textsf{NewSample}$. Recall that $T = \lceil 1000\log \frac{n}{\eps} \rceil$ is the parameter in \Cref{alg-sample}.
\begin{definition}[$\textsf{NewSample}$]\label{def-new-sample}
The algorithm $\textsf{NewSample}$ is the same as the  $\textsf{Sample}$ in \Cref{alg-sample}. The only difference is that before \Cref{line-last-sample}, $\textsf{NewSample}$ computes the value 
\begin{align}\label{eq-def-pK}
  p_K \defeq \frac{1-\frac{\epsilon^{200}}{n^{200}}}{1 - (1-\frac{R_u}{4p_0})^T}.
\end{align}
If $0 \leq  p_K \leq 1$, then independently sample $K \in \{0,1\}$ such that $\Pr[K = 1] = p_K$; otherwise, let $K = 0$.
\end{definition}

We remark that in the above definition, $K$ is sampled using independent randomness. Formally, let $\+D_u$ be the inside randomness of $\textsf{NewSample}$. We partition $\+D_u$ into two disjoint random strings $\+D_u^{(1)}$ and $\+D_u^{(2)}$. We use $\+D_u^{(1)}$ to simulate all steps in \Cref{alg-sample} and use $\+D_u^{(2)}$ to sample $K$. 

The value $R_u$ is the exact $u$-$t$ reliability in graph $G_u$. Indeed, we cannot compute the exact value of $R_u$ in polynomial time. We only use the algorithm $\textsf{NewSample}$ in analysis, and do not need to implement this algorithm. $\textsf{NewSample}$ draws a random variable $K$ but never uses it at all. Its sole purpose is to further refine the probability space. Thus, the following observation is straightforward to verify.
\begin{observation}\label{ob-new-sample}
Given the same input, the outputs of two algorithms $\textsf{NewSample}$ and $\textsf{Sample}$ follow the same distribution.
\end{observation}
A natural question here is that why do we even define $\textsf{NewSample}$? By \Cref{ob-new-sample}, we can focus only on $\textsf{NewSample}$ in later analysis (in particular, the analysis of the correctness of our algorithm).
$\textsf{NewSample}$ has one additional random variable $K \in \{0,1\}$, which helps defining the event $\+C$ below.

Similarly, we can defined a modified version of $\textsf{NewSample}$ such that we use the oracle $\+P$ to compute $p_0,c_0$ and $c_1$. 
$\textsf{NewSample}$ also generates the same random subgraph $H = H(\+X_{\+P},\+D_u^{(1)})$ and the same random indicator $F = F(\+X_{\+P},\+D_u^{(1)})$ as \textsf{Sample}.
In addition, it  generates a new random variable $K = K(\+X_{\+P},\+D_u^{(2)})$.
We define the following event $\+C$ for $\textsf{NewSample}$
\begin{align}\label{eq-def-C}
	\+C: \quad F(\+X_{\+P},\+D_u^{(1)}) = 1 \land K(\+X_{\+P},\+D_u^{(2)}) = 1.
\end{align}

\begin{lemma}\label{lemma-new-sample}
Suppose $\+P$ satisfies the conditions in~\eqref{eq-r1},~\eqref{eq-r2} and~\eqref{eq-r3}. Then with probability $1$, $0 \leq p_K \leq 1$.  Furthermore, it holds that
\begin{itemize}
	\item $\+C$ is independent from $\+X_{\+P}$, which implies that $\+C$ depends only on $\+D_u$;
	\item $Pr_{\+D_u}[\+C] = 1 - \eps^{200}/n^{200}$;
	\item conditional on $\+C$, \text{NewSample} does not crash and outputs an independent sample $H(\+X_{\+P},\+D_u^{(1)}) \sim \pi_u$.
\end{itemize}
\end{lemma}
\begin{proof}
	Suppose $\+P$ satisfies the conditions in~\eqref{eq-r1},~\eqref{eq-r2} and~\eqref{eq-r3}. 
We first fix $\+X_{\+P} = x_{\+P}$ for an arbitrary $x_{\+P} \in \Omega_{\+P}$.
By the same analysis as in \Cref{lemma-sample}, in each of the repeat-until loop, the probability $f$ of $F = 1$ is 
\begin{align*}
  \frac{1}{16}\leq f = \sum_{H \in \Omega_u} p(H)\cdot\frac{w_u(H)}{4p(H)p_0} =  \frac{R_u}{4p_0} \leq 1.
\end{align*}
As the repeat-until loop is repeated independently for at most $T$ times until $F(\+X_{\+P},\+D_u)=1$, we have
\begin{align*}
	\Pr[F(\+X_{\+P},\+D_u) = 1 \mid \+X_{\+P} = x_{\+P}] = 1- \tp{1 - \frac{R_u}{4p_0}}^T \geq 1  - \frac{\eps^{200}}{n^{200}}. 
\end{align*}
Hence, $0\leq p_K \leq 1$. 
Next, note that given $\+X_{\+P} = x_{\+P}$, the value of $p_K$ is fixed and $K$ is sampled independently. We have that $K(\+X_{\+P},\+D_u^{(2)})$ and $F(\+X_{\+P},\+D_u^{(1)})$ are independent conditional on $\+X_{\+P} = x_{\+P}$. This implies 
\begin{align*}
	\Pr[\+C \mid \+X_{\+P} = x_{\+P}] &= \Pr[F(\+X_{\+P},\+D_u^{(1)}) = 1 \mid \+X_{\+P} = x_{\+P}]\Pr[K(\+X_{\+P},\+D_u^{(2)}) = 1 \mid \+X_{\+P} = x_{\+P}]\\
	&= \tp{1- \tp{1 - \frac{R_u}{4p_0}}^T}p_K = 1 - \frac{\eps^{200}}{n^{200}}. 
\end{align*}
The probability $1 - \eps^{200} / n^{200}$ is independent from $x_{\+P}$. Hence, the event $\+C$ is independent from $\+X_{\+P}$.

Finally, we analyze the distribution of $H$ conditional on $\+C$. We first condition on $\+X_{\+P} = x_{\+P}$. If we further conditional on $F(\+X_{\+P},\+D_u^{(1)}) = 1$, the same analysis as in \Cref{lemma-sample} shows that the algorithm does not crash and $H(\+X_{\+P},\+D_u^{(1)}) \sim \pi_u$. Note that $K(\+X_{\+P},\+D_u^{(2)})$ is sampled independently with a fixed probability~$p_K$ (since $\+X_{\+P} = x_{\+P}$ has been fixed). Hence, $K(\+X_{\+P},\+D_u^{(2)})$ is independent from both $F(\+X_{\+P},\+D_u^{(1)})$ and $H(\+X_{\+P},\+D_u^{(1)})$ conditional on $\+X_{\+P} = x_{\+P}$. We have
\begin{align*}
	 H(\+X_{\+P},\+D_u^{(1)})|_{\+X_{\+P} = x_{\+P} \land \+C}   \equiv   H(\+X_{\+P},\+D_u^{(1)})|_{\+X_{\+P} = x_{\+P} \land F(\+X_{\+P},\+D_u^{(1)}) = 1}  \sim \pi_u,
\end{align*}
where we use $X \equiv Y$ to denote that two random variables $X$ and $Y$ have the same distribution.
Note that the distribution $\pi_u$ on the RHS is independent from $x_{\+P}$. Summing over $x_{\+P}\in\Omega_{\+P}$ gives that conditioned on $\+C$, the output $H = H(\+X_{\+P},\+D_u^{(1)}) \sim \pi_u$.
\end{proof}

\subsection{Analysis of \textsf{ApproxCount}}
Now we turn our attention to $\textsf{ApproxCount}(V,E,\Lambda,(\widetilde{R}_w,S_w)_{w \in \partial \Lambda})$, where $G=(V,E)$ is a DAG and $t \not\in \Lambda$.
Recall that for any $w \in V$, the graph $G_w = G[V_w]$, where $V_w$ contains all vertices $v$ satisfying $\rch{w}{v}{G}$ and $\rch{v}{t}{G}$. Let $R_w$ be the $w-t$ reliability in $G_w$. Let $S^{\-{ideal}}_w$ be a multi-set of $\ell$ \emph{independent} and \emph{prefect} samples from~$\pi_w$.
Recall that $\ell_0$ and $B$ are parameters in \textsf{ApproxCount}, \Cref{alg-count}, and $d=\abs{\partial\Lambda}$.
In the next lemma, we assume $(\widetilde{R}_w)_{w \in \partial \Lambda}$ is fixed and $(S_w)_{w \in \partial \Lambda}$ is random.
When \textsf{ApproxCount} is called, we use $\+D(V,E,\Lambda)$ to denote the internal randomness in the execution of \textsf{ApproxCount}.
\begin{lemma}\label{lemma-ac}
Suppose the following conditions are satisfied
\begin{itemize}
	\item for all $w \in \partial \Lambda$, $\rch{w}{t}{G}$;
	\item for any $w \in \partial \Lambda$, $1-\eps_0 \leq \frac{\widetilde{R}_w}{R_w} \leq 1+\eps_0 $ for some $\eps_0 < 1/2$;
    \item $\DTV{(S_w)_{w \in \partial \Lambda}}{(S^{\-{ideal}}_w)_{w \in \partial \Lambda}} \leq \delta_0$.
\end{itemize}
Then with probability at least $1 - \delta_0 - 2^{-B/30}$, it holds that 
\begin{align*}
1 - \eps_0 - \frac{2}{\sqrt{\ell_2}}	\leq \frac{\widetilde{R}_{\Lambda}}{R_{\Lambda}} \leq 1 + \eps_0 + \frac{2}{\sqrt{\ell_2}},
\end{align*}
where the probability is taken over the input randomness of $(S_w)_{w \in \partial \Lambda}$ and the independent randomness $\+D(V,E,\Lambda)$ inside the \emph{\textsf{ApproxCount}} algorithm. The running time of \emph{\textsf{ApproxCount}} is $O(n\ell(|V|+|E|))$.
\end{lemma}
\begin{proof}
If $t \in \Lambda$, then \Cref{alg-count} returns $\widetilde{R}_w = {R}_w = 1$ in \Cref{line-trival}. If $t \not\in \Lambda$ and $\partial\Lambda = \emptyset$, then \Cref{alg-count} returns $\widetilde{R}_w = {R}_w = 0$ in \Cref{line-trival}. In the following, we assume $t \not\in \Lambda$ and $\partial \Lambda \neq \emptyset$. By the first condition, we have $d \geq 1$ in \Cref{alg-count} and all distributions $\pi^{(i)}_{\Lambda}$ for $i \in [d]$ are well-defined.

By \Cref{lemma-expand-R} and the assumption in this lemma, all $\widetilde{R}_{\Lambda}^{(i)}$ computed in \Cref{line-R-tilde} satisfy
\begin{align}\label{eq-bound-RR}
	 1 - \eps_0 \leq \frac{\widetilde{R}^{(i)}_{\Lambda}}{R^{(i)}_{\Lambda}} \leq 1 + \eps_0.
\end{align}
Suppose \textsf{ApproxCount} uses perfect samples from $(S^{\-{ideal}}_w)_{w \in \partial \Lambda}$.
Consider the first call on the subroutine \textsf{Estimate} (\Cref{alg-es}).
In the first call, the parameter $T = \ell_1 = 400n$. Note that $\ell_{\-{es}} = \ell_1 = 400n$. Hence, the condition in \Cref{line-break} of \Cref{alg-es} cannot be satisfied.
By \Cref{lemma-expand-R}, $H'$ obtained in \Cref{line-H'} of \Cref{alg-es} is a perfect sample from $\pi_\Lambda^{(i)}$.
For any $k \in [T]$, we have
\begin{align*}
	\Ex[Z_{\-{es}}^{(k)}] = \sum_{i=1}^d\frac{ \widetilde{R}^{(i)}_\Lambda }{\sum_{t\in [d]}\widetilde{R}_\Lambda^{(t)}} \cdot \sum_{H' \in \Omega_{\Lambda}^{(i)} } \underbrace{\frac{w_{\Lambda}(H')}{R^{(i)}_\Lambda}}_{=\pi_\Lambda^{(i)}(H')} \cdot \boldsymbol{1}\left[H' \in \Omega^{(i)}_\Lambda \land \forall t < i, H' \notin \Omega^{(t)}_{\Lambda}\right]
\end{align*}
By~\eqref{eq-bound-RR} and a calculation similar to that in~\eqref{eqn:expectation-lb}, we have 
\begin{align}\label{eq-ex-bound}
(1-\eps_0) \frac{R_{\Lambda}}{\sum_{i=1}^d \widetilde{R}_{\Lambda}^{(i)}}	 \leq \Ex[Z_{\-{es}}^{(k)}] \leq (1+\eps_0)\frac{R_{\Lambda}}{\sum_{i=1}^d \widetilde{R}_{\Lambda}^{(i)}}	
\end{align}
Using~\eqref{eq-bound-RR}, we have 
\begin{align}\label{eqn:Z-lambda}
\Ex[Z_{\-{es}}^{(k)}] \geq \frac{1-\eps_0}{1+\eps_0} \cdot \frac{R_{\Lambda}}{\sum_{i=1}^d \widetilde{R}_{\Lambda}^{(i)} } \geq \frac{1}{4d}.	
\end{align}
Also recall that
\begin{align}\label{eq-def-Q}
	Z_{\-{es}} = \frac{1}{T}\sum_{k=1}^{T}Z_{\-{es}}^{(k)}.
\end{align}
Then $\Var{}{Z_{\-{es}}}=\frac{\Var{}{Z_{\-{es}}^{(k)}}}{T}$, and by \eqref{eqn:Z-lambda}, $\Ex[Z_{\-{es}}]\geq \frac{1}{4d}$.
By Chebyshev's inequality and for any $k \in [T]$, 
\begin{align*}
    &\Pr\left[ \abs{Z_{\-{es}} - \Ex[Z_{\-{es}}] }\geq \frac{10\sqrt{d}}{\sqrt{T}} \Ex[Z_{\-{es}}]  \right] 
    \le \frac{T}{100d}\cdot \frac{\Var{}{Z_{\-{es}}}}{(\Ex[Z_{\-{es}}])^2} = \frac{T}{100d}\cdot \frac{\Var{}{Z_{\-{es}}^{(k)}}/T}{\left(\Ex[Z_{\-{es}}^{(k)}]\right)^2} \\
   =\; & \frac{1}{100d}\left( \frac{\Ex\left[\left(Z_{\-{es}}^{(k)}\right)^2\right]}{\left(\Ex[Z_{\-{es}}^{(k)}]\right)^2} - 1 \right) 
    = \frac{1}{100d}\left( \frac{\Ex\left[Z_{\-{es}^{(k)}}\right]}{\left(\Ex[Z_{\-{es}^{(k)}}]\right)^2} - 1 \right) \le \frac{1}{25},
\end{align*}
where the last inequality is due to \eqref{eqn:Z-lambda}. Since $T = \ell_1 =  400n \geq 100d$ in the first call on \Cref{alg-es}, the random variable $\widehat{Z}_{\Lambda}^{(j)}$ in \Cref{alg-count} satisfies
\begin{align*}
	\Pr \left[ \frac{1}{2}  \leq \frac{\widehat{Z}_{\Lambda}^{(j)}}{\Ex[Z_{\-es}]} \leq 2 \right] \geq \frac{24}{25}.
\end{align*}

Next, we analyse the second call on \Cref{alg-es} conditional on $\frac{1}{2}  \leq \frac{\widehat{Z}_{\Lambda}^{(j)}}{\Ex[Z_{\-es}]} \leq 2 $. 
By \eqref{eqn:Z-lambda}, $\Ex[Z_{\-es}] \geq \frac{1}{4d} \geq \frac{1}{4n}$.
In this case, the parameter $T$ in \Cref{alg-es} satisfies 
\begin{align}\label{eq-T-analysis}
\frac{25 \ell_2}{\Ex[Z_{\-{es}}]} \leq  T = 25 \ell_2 \min\{ 2/ \widehat{Z}_{\Lambda}^{(j)},4n\}	 \leq  \min \left\{\frac{100 \ell_2}{\Ex[Z_{\-{es}}]}, 100\ell_2n\right\}.
\end{align}
Consider the second call of \Cref{alg-es}. 
Note that  $\ell_{\-{es}} = 500\ell_2$ in the second round.
Let us first assume that $\ell_{\-{es}} = \infty$, which means each $S^{\-{es}}_{u_i}$ contains infinitely many perfect samples. 
We first analyse the algorithm in this ideal situation and then compare the real algorithm (where $\ell_{\-{es}} = 500\ell_2$) with this ideal algorithm.  
Note that if  $\ell_{\-{es}} = \infty$, then the condition in \Cref{line-break} of \Cref{alg-es} cannot be triggered. By a similar analysis,
\begin{align*}
 &\Pr\left[ \abs{Z_{\-{es}} - \Ex[Z_{\-{es}}] } \geq \frac{1}{\sqrt{\ell_2}} \Ex[Z_{\-{es}}] \right] \leq \ell_2\cdot \frac{\Var{}{Z_{\-{es}}}}{(\Ex[Z_{\-{es}}])^2} = \frac{\ell_2}{T} \cdot \frac{\Var{}{Z_{\-{es}}^{(k)}}}{\left(\Ex[Z_{\-{es}}^{(k)}]\right)^2}\\
 \leq \,& \frac{1}{25}\cdot \frac{\Var{}{Z_{\-{es}}^{(k)}}}{\Ex[Z_{\-{es}}^{(k)}]} \leq \frac{1}{25}\cdot \frac{\Ex\left[\tp{Z_{\-{es}}^{(k)}}^2\right]}{\Ex[Z_{\-{es}}^{(k)}]} = \frac{1}{25}.
\end{align*} 
We then show that the following result holds at the end of this ideal algorithm ($\ell_{\-{es}} = \infty$)
\begin{align*}
	\Pr\left[\exists i \in [d], \text{ s.t. } c_i > 500\ell_2 \right] \leq \frac{1}{25}.
\end{align*}
Fix an index $i \in [d]$. \Cref{alg-es} has $T$ iterations in total. For any $k \in [T]$, let $X_k \in \{0,1\}$ indicate whether $i$ is picked in \Cref{line-pick-sample}. Note that all $X_k$'s are independent random variables. Let $X = \sum_{k=1}^T X_k$. Then
\begin{align*}
	\Ex[X] = T \frac{ \widetilde{R}^{(i)}_\Lambda }{\sum_{t\in [d]}\widetilde{R}_\Lambda^{(t)}} \leq \frac{100\ell_2}{\Ex[Z_{\-{es}}]}\frac{ \widetilde{R}^{(i)}_\Lambda }{\sum_{t\in [d]}\widetilde{R}_\Lambda^{(t)}},
\end{align*}
where the inequality holds by the upper bound in~\eqref{eq-T-analysis}. Since $\Ex[Z_{\-{es}}] = \Ex[Z_{\-{es}}^{(k)}]$, we can use the lower bound in \eqref{eq-ex-bound} and the upper bound in~\eqref{eq-bound-RR} to obtain
\begin{align*}
\Ex[X] \leq \frac{100\ell_2}{1- \epsilon_0}\cdot \frac{\widetilde{R}^{(i)}_\Lambda}{R_\Lambda} \leq 100 \ell_2 \cdot \frac{1+\epsilon_0}{1-\epsilon_0} \cdot \frac{R^{(i)}_\Lambda}{R_\Lambda}  \overset{(\ast)}{\leq} 100 \ell_2 \cdot \frac{1+\epsilon_0}{1-\epsilon_0}  \leq 300\ell_2,
\end{align*}
where inequality~$(\ast)$ uses the fact $R^{(i)}_\Lambda \leq R_\Lambda$. 
This is because $R^{(i)}_\Lambda = \sum_{H \in \Omega_\Lambda^{(i)}}w_\Lambda(H)$, $R_\Lambda = \sum_{H \in \Omega_\Lambda}w_\Lambda(H)$ and $\Omega_\Lambda^{(i)} \subseteq \Omega_{\Lambda}$ (by \Cref{lemma-union}).
Note that $T \leq 100\ell_2 n$ and $\ell_2 \geq 10^4 n^2 \max\{m^2,\epsilon^{-2}\}$.
Using Hoeffding inequality on~$X$ yields 
\begin{align*}
	\Pr[X > 500\ell_2 ] \leq \Pr[X > \Ex[X] + 200 \ell_2] \leq \exp\tp{-\frac{200^2 \ell_2^2}{T}} \leq \frac{1}{25n}.
\end{align*}
By a union bound over all $i \in [d]$, we have
\begin{align*}
	\Pr\left[\exists i \in [d], \text{ s.t. } c_i > 500\ell_2 \right] \leq \frac{1}{25}.	
\end{align*}
We can couple the ideal algorithm ($\ell_{\-{es}} = \infty$) with the real algorithm ($\ell_{\-{es}} = 500\ell_2$) such that 
if the above bad event does not occur, the two algorithms output the same value. 
Hence, the random variable $\widetilde{Z}_{\Lambda}^{(j)}$ in \Cref{alg-count} satisfies
\begin{align*}
	&\Pr \left[ 1 - \frac{1}{\sqrt{\ell_2}} \leq \frac{\widetilde{Z}_{\Lambda}^{(j)}}{\Ex[Z_{\-es}]} \leq 1+ \frac{1}{\sqrt{\ell_2}}  \right]\\
	 \geq\,& 	\Pr \left[ \frac{1}{2}  \leq \frac{\widehat{Z}_{\Lambda}^{(j)}}{\Ex[Z_{\-es}]} \leq 2 \right] \Pr\left[  1 - \frac{1}{\sqrt{\ell_2}} \leq \frac{\widetilde{Z}_{\Lambda}^{(j)}}{\Ex[Z_{\-es}]} \leq 1+ \frac{1}{\sqrt{\ell_2}} \; \middle | \; \frac{1}{2}  \leq \frac{\widehat{Z}_{\Lambda}^{(j)}}{\Ex[Z_{\-es}]} \leq 2  \right]\\
	 \geq\,& \frac{24}{25}\tp{1 - \frac{1}{25}-\frac{1}{25}} \geq \frac{3}{4}.
\end{align*}
Combining~\eqref{eq-ex-bound} with the above inequality, we know that with probability at least $3/4$, the random variable $Q^{(j)}_\Lambda$ in \Cref{alg-count} satisfies 
\begin{align*}
1 - \eps_0 - \frac{2}{\sqrt{\ell_2}} \leq \tp{1-\frac{1}{\sqrt{\ell_2}}}(1-\eps_0) \leq \frac{Q^{(j)}_\Lambda}{R_\Lambda} \leq \tp{1+\frac{1}{\sqrt{\ell_2}}}(1+\eps_0) \leq 1 + \eps_0 + \frac{2}{\sqrt{\ell_2}}.
\end{align*}
Since $\widetilde{R}_{\Lambda}$ is the median of $B$ values $Q^{(j)}_\Lambda$, the success probability is boosted from $3/4$ to $1 - 2^{-B/30}$ by the Chernoff bound.

Finally, the algorithm actually uses the samples from $(S_w)_{w \in \partial \Lambda}$. Consider an optimal coupling between the real algorithm with the algorithm using ideal samples from $(S^{\-{ideal}}_w)_{w \in \partial \Lambda}$. Due to the assumption that $\DTV{(S_w)_{w \in \partial \Lambda}}{(S^{\-{ideal}}_w)_{w \in \partial \Lambda}} \leq \delta_0$ and \Cref{lem:coupling}, 
the two algorithms output the same answer with probability at least $1 - \delta_0$. Hence, $1 - \eps_0 - \frac{2}{\sqrt{\ell_2}}	\leq \frac{\widetilde{R}_{\Lambda}}{R_{\Lambda}} \leq 1 + \eps_0 + \frac{2}{\sqrt{\ell_2}}$ with probability at least $1 - \delta_0 - 2^{-B/30}$.

The running time of \Cref{alg-count} is dominated by the second call on \textsf{Estimate} with parameter $T = O(\ell_2 n)$. 
In \Cref{alg-es}, the running time is dominated by the time spent on \Cref{line-test}.
We can find all the vertices that can reach $t$ in graph $H'$ in time $O(|E|+|V|)$ (first inverse the direction of all edges and then run a BFS starting from $t$). 
We can then compute $Z^{(k)}_{\-{es}}$ in time $O(|V|)$.
The total running time is 
\begin{align*}
	O(BT(|V|+|E|)) = O(B\ell_2 n (|V|+|E|)) = O(n\ell (|V|+|E|)). &\qedhere
\end{align*}
\end{proof}

\Cref{lemma-ac} treats \textsf{ApproxCount} as a standalone algorithm. However, in our main algorithm, we use \textsf{ApproxCount} as a subroutine. 
We need to make sure that the inpputs are consistent every time \textsf{ApproxCount} is called.
Recall that $G=(V,E)$ denotes the input graph of \Cref{alg-main}. Every time when $\textsf{ApproxCount}$ is evoked, its input includes a subset of vertices $V_0 \subseteq V$, a subset of edges $E_0 \subseteq E$, a subset of vertices $\Lambda_0$ and $(\widetilde{R}_w,S_w)_{w \in \partial \Lambda_0}$.
Recall that $\partial\Lambda_0 = \{w \in V_0 \setminus \Lambda_0 \mid \exists w \in \Lambda_0 \text{ s.t. } (w',w)\in E_0\}$.
The properties we need are the following.
\begin{lemma}\label{lem-input-condition}
If $\textsf{ApproxCount}$ is evoked with input $V_0 \subseteq V$, $E_0 \subseteq E$, $\Lambda_0 \subseteq V_0$ and $(\widetilde{R}_w,S_w)_{w \in \partial \Lambda_0}$, then
\begin{itemize}
  \item for all $w \in \partial \Lambda_0$, $E_w\subseteq E_0$, where $E_w$ is the edge set of $G_w$;
  \item for all $w \in \partial \Lambda_0$, $(\widetilde{R}_w,S_w)$ has already been computed.
\end{itemize}
\end{lemma}

\begin{remark}\label{remark-welldfine}
  \Cref{lem-input-condition} guarantees that for any input $(V_0,E_0,\Lambda_0)$ of \Cref{alg-count}, for any $w \in \partial \Lambda_0$, $G^0_w = G_w$, where $G^0=(V_0,E_0)$. This implies that $\rch{w}{t}{{G^0}}$ and all distributions in~\eqref{eqn:pi^i_Lambda} are well-defined.
\end{remark}

\begin{proof}[Proof of \Cref{lem-input-condition}]
Note that \textsf{ApproxCount} can be evoked either in \Cref{line-count} of \Cref{alg-main} or in \Cref{line-c1} and \Cref{line-c2} in \Cref{alg-sample}.

If \textsf{ApproxCount} is  evoked  by \Cref{alg-main}, then $V_0 = V_{v_k}$, $E_0=E_{v_k}$ and $\Lambda_0 = \{v_k\}$. For any $w \in V_{v_k} \setminus \{v_k\}$, $G_w$ is a subgraph of $G_{v_k} = (V_0,E_0)$ and the first property holds. The second property holds because $v_k \prec w$ for all $w \in V_{v_k} \setminus \{v_k\}$.
	
Suppose next that \textsf{ApproxCount} is evoked  by \Cref{alg-sample}.
For the first property, by~\Cref{lemma-loop}, at the beginning of every while-loop, for any $w \in \partial \Lambda$, $E_1 \cap E_w = \emptyset$, which implies $E_w \subseteq E_2$.
In other words, the property holds at the beginning of the loop with $V_0=V_u$, $E_0=E_2$, and $\Lambda_0=\Lambda$.
Now consider \Cref{line-c1} and \Cref{line-c2} separately.
\begin{itemize}
  \item Suppose \textsf{ApproxCount} is called in \Cref{line-c1}.
    In this case, comparing to the beginning of the loop,
    $\partial \Lambda$ can only be smaller,
    and the edge $(w',w^*)$ is removed from $E_2$, where $w'\in\Lambda$.
    If the property does not hold,
    then there is some $w''\in\partial\Lambda$ such that $E_{w''}$ is not contained in the current $E_2$.
    This means that $(w',w^*)\in E_{w''}$,
    which implies that $w''\prec w^*$. This contradicts how $w^*$ is chosen.
  \item Next consider \Cref{line-c2}.
    By \Cref{cor-loop}, $\Lambda_1 = \Lambda \cup \{w^*\}$.
    Note that $\partial\Lambda_1 = \partial\Lambda\cup\Gamma_{\-{out}}(w^*) \setminus \{w^*\}$, where $\Gamma_{\-{out}}(w^*) = \{v \in V_u \setminus \Lambda \mid (w^*,v) \in E_2\}$.
    The property holds for all vertices in $\partial\Lambda\setminus\{w^*\}$ by the previous case.
    For $w''\in \Gamma_{\-{out}}(w^*) \setminus \partial \Lambda $, the removal of the edge $(w',w^*)$ does not affect $G_{w''}$.
    Thus the property also holds.
\end{itemize}


The second property holds because $u \prec w$ for all $w \in V_u \setminus \{u\}$.
\end{proof}

In \Cref{alg-main}, each $(\widetilde{R}_w,S_w)$ is computed with respect to $G_w$. \Cref{lem-input-condition} together with \Cref{lemma-ac} shows that for every instance of $\textsf{ApproxCount}$ evoked by \Cref{alg-main}, we can reuse all $(\widetilde{R}_w,S_w)$ computed before.

We still need to take care of the case when \Cref{alg-count} is called by \Cref{alg-sample}.
This requires a generalised version of \Cref{lemma-ac}. 
Again, let $G=(V,E)$ be the input of \Cref{alg-main} and $v_1,v_2,\ldots,v_n \in V$, where $v_1 = s$ and $v_n = t$, be the topological ordering in \Cref{alg-main}. For $i$ from $n$ to $1$, \Cref{alg-main} compute $\widetilde{R}_{v_i}$ and a multi-set $S_{v_i}$ of $\ell$ random samples step by step. 
For any fixed $i$, we view each $(\widetilde{R}_{v_j},S_{v_{j}})_{j > i}$ as a random variable following a joint distribution.

Every time when \textsf{ApproxCount} (described in \Cref{alg-count}) is evoked by \Cref{alg-main} or \Cref{alg-sample}, its input includes a subset of vertices $V_0 \subseteq V$ with $t \in V_0$, a subset of edges $E_0 \subseteq E$, a subset of vertices $\Lambda_0 \subseteq V_0$ and $(\widetilde{R}_w,S_w)_{w \in \partial \Lambda_0}$, where $\partial\Lambda_0 = \{w \in V_0 \setminus \Lambda_0 \mid \exists w' \in \Lambda_0 \text{ s.t. } (w',w)\in E_0\}$. For any $i$, define $\Phi_i$ as a set of tuples $(V_0,E_0,\Lambda_0)$ such that
\begin{itemize}
	\item $(V_0,E_0,\Lambda_0) \in 2^V \times 2^E \times 2^{V_0}$;
	\item for all $w \in \partial \Lambda_0$, $E_w\subseteq E_0$, where $E_w$ is the edge set of $G_w$;
	\item $\partial\Lambda_0 \subseteq \{v_{i+1},\ldots,v_{n}\}$.
\end{itemize}

By \Cref{lem-input-condition} and the way \Cref{alg-main} works, $\Phi_i $ contains all possible inputs  of \textsf{ApproxCount} when we compute $\widetilde{R}_{v_i}$ and $S_{v_i}$ (including the recursive calls). 
For any $(V_0,E_0,\Lambda_0) \in \Phi_i$, let $R(V_0,E_0,\Lambda_0)$ denote the $\Lambda_0 -t$ reliability in the graph $G_0 = (V_0, E_0)$, where every edge $e \in E_0$ fails independently with probability $q_e$.
Suppose the random tuples $(\widetilde{R}_{v_j},S_{v_{j}})_{j > i}$ have been  generated by \Cref{alg-main}.
If we run $\textsf{ApproxCount}$ on $(V_0,E_0,\Lambda_0)$ and $(\widetilde{R}_w,S_w)_{w \in \partial  \Lambda_0}$ (note that $\partial \Lambda_0 \subseteq \{v_{i+1},\ldots,v_n\}$), it will return a random number $\widetilde{R}(V_0,E_0,\Lambda_0)$, where the randomness comes from the input randomness of $(\widetilde{R}_{v_j},S_{v_{j}})_{j > i}$ and the independent randomness $\+D(V_0,E_0,\Lambda_0)$ inside $\textsf{ApproxCount}$.
Our implementation makes sure that any \textsf{ApproxCount} is evoked for every $(V_0,E_0,\Lambda_0)$ at most once.
Hence, there is a unique random variable $\widetilde{R}(V_0,E_0,\Lambda_0)$ for each $(V_0,E_0,\Lambda_0)$.

We have the following generalised version of \Cref{lemma-ac}.
Recall that $S^{\-{ideal}}_w$ is a set of $\ell$ independent perfect samples from the distribution $\pi_w$.
\begin{lemma}\label{lemma-ac-gen}
Given the random tuples $(\widetilde{R}_{v_j},S_{v_{j}})_{j > i}$ such that the following two conditions are satisfied
\begin{itemize}
  \item for all $j > i$, $1-\eps_0 \leq \frac{\widetilde{R}_{v_j}}{R_{v_j}} \leq 1+\eps_0 $ for some $\eps_0 < 1/2$;
  \item $\DTV{(S_{v_j})_{j > i}}{(S^{\-{ideal}}_{v_j})_{j > i}} \leq \delta_0$.
\end{itemize}
Let $\Phi \subseteq \Phi_i$.
Then with probability at least $1 - \delta_0 - |\Phi|2^{-B/30}$, it holds that 
\begin{align}\label{eq-lemma-good}
\forall (V_0,E_0,\Lambda_0) \in \Phi,\quad
1 - \eps_0 - \frac{2}{\sqrt{\ell_2}}	\leq \frac{\widetilde{R}(V_0,E_0,\Lambda_0)}{R(V_0,E_0,\Lambda_0)} \leq 1 - \eps_0 - \frac{2}{\sqrt{\ell_2}},
\end{align}
where the probability is taken over the randomness of $(\widetilde{R}_{v_j},S_{v_{j}})_{j > i}$ and the independent randomness of $\+D(V_0,E_0,\Lambda_0)$ for $(V_0,E_0,\Lambda_0) \in \Phi$. 
\end{lemma}


\begin{proof}[Proof of \Cref{lemma-ac-gen}]
  Strictly speaking, all $(\widetilde{R}_{v_j})_{j > i}$ are random variables, and the first condition means that the event for $(\widetilde{R}_{v_j})_{j > i}$ holds with probability $1$. 
  We will actually prove a stronger result. 
  Namely, the lemma holds with probability at least $1 - \delta_0 - |\Phi|2^{-B/30}$, where the probability is taken over the randomness of $(S_{v_{j}})_{j > i}$ and the independent randomness of $\+D(V_0,E_0,\Lambda_0)$ for $(V_0,E_0,\Lambda_0) \in \Phi$, for any fixed values of $\widetilde{R}_{v_j}$ as long as the first condition is met.


  Note that for any $(V_0,E_0,\Lambda_0) \in \Phi$, for all $w \in \partial \Lambda_0$,  $E_w\subseteq E_0$, where $E_w$ is the edge set of $G_w$, and $\partial\Lambda_0 = \{w \in V_0 \setminus \Lambda_0 \mid \exists w \in \Lambda_0 \text{ s.t. } (w',w)\in E_0\}$, which also implies that $w$ can reach $t$ in the graph $(V_0,E_0)$. 
  The assumption in this lemma implies the assumption in \Cref{lemma-ac}, and the proof of this lemma is similar to the proof of \Cref{lemma-ac}.

Again, the  cases where $t \in \Lambda_0$ and  where $\partial \Lambda_0 = \emptyset$ are trivial. The main case is when $t\not\in\Lambda_0$ and $\partial\Lambda_0\neq\emptyset$. We first use $(S^{\-{ideal}}_{v_j})_{j > i}$ to run the algorithm. 
Let use denote the output of the algorithm by $\widetilde{R}^{\-{ideal}}(V_0,E_0,\Lambda_0)$.
By the same argument as the one for \Cref{lemma-ac}, for any $(V_0,E_0,\Lambda_0) \in \Phi$, with probability at least $1-2^{-B/30}$,
\begin{align*}
1 - \eps_0 - \frac{2}{\sqrt{\ell_2}}		\leq \frac{\widetilde{R}^{\-{ideal}}(V_0,E_0,\Lambda_0)}{R(V_0,E_0,\Lambda_0)} \leq 1 + \eps_0 + \frac{2}{\sqrt{\ell_2}}.
\end{align*}
By a union bound over all  $(V_0,E_0,\Lambda_0) \in \Phi$, we have that with probability at least $1-|\Phi|2^{-B/30}$,
\begin{align}\label{eq-R-good}
\forall (V_0,E_0,\Lambda_0) \in \Phi, \quad  1 - \eps_0 - \frac{2}{\sqrt{\ell_2}}		\leq \frac{\widetilde{R}^{\-{ideal}}(V_0,E_0,\Lambda_0)}{R(V_0,E_0,\Lambda_0)} \leq 1 + \eps_0 + \frac{2}{\sqrt{\ell_2}}.
\end{align}

Then we show the lemma using an optimal coupling between ${(S_{v_j})_{j > i}}$ and ${(S^{\-{ideal}}_{v_j})_{j > i}} $. To be more precise, we first sample  ${(S_{v_j})_{j > i}}$ and ${(S^{\-{ideal}}_{v_j})_{j > i}} $ from their optimal coupling, then by \Cref{lem:coupling} we have
\begin{align}\label{eq-couple-bound}
	\Pr\left[\forall j > i, S_{v_j} = S^{\-{ideal}}_{v_j} \right] \geq 1 - \delta_0.
\end{align}
Next, we sample all $\+D = (\+D(V_0,E_0,\Lambda_0))_{(V_0,E_0,\Lambda)\in \Phi}$. 
When we use ${(S_{v_j})_{j> i}}$ and $\+D$ to run \textsf{ApproxCount} on all of $(V_0,E_0,\Lambda)\in \Phi$, we obtain an output vector $\widetilde{R} = (\widetilde{R}(V_0,E_0,\Lambda_0))_{(V_0,E_0,\Lambda)\in \Phi}$. 
Similarly, denote by $\widetilde{R}^{\-{ideal}} = (\widetilde{R}^{\-{ideal}}(V_0,E_0,\Lambda_0))_{(V_0,E_0,\Lambda)\in \Phi}$ the output vector when we use ${(S^{\-{ideal}}_{v_j})_{j > i}}$ and $\+D$ to run \textsf{ApproxCount}.
Define two good events
\begin{itemize}
	\item $A_1$: $\widetilde{R}^{\-{ideal}} = \widetilde{R} $. By~\eqref{eq-couple-bound}, $\Pr[A_1] \geq 1 - \delta_0$;
	\item $A_2$: ~\eqref{eq-R-good} holds for $\widetilde{R}^{\-{ideal}}$. We know $\Pr[A_2] \geq 1-|\Phi|2^{-B/30}$.
\end{itemize}
If both $A_1$ and $A_2$ occur, then~\eqref{eq-lemma-good} holds. The probability is 
\begin{align*}
	\Pr[A_1 \land A_2] = 1 - \Pr[\overline{A_1} \lor \overline{A_2}] \geq 1 - \Pr[\overline{A_1}] - \Pr[\overline{A_2}] \geq 1 - \delta_0 -|\Phi|2^{-B/30}. &\qedhere
\end{align*}  
\end{proof}

Note that \Cref{lemma-ac-gen} cannot be obtained by simply applying \Cref{lemma-ac} with a union bound, 
as that will result in a failure probability of $\abs{\Phi}\left( \delta_0+2^{-B/30} \right)$ instead of $\delta_0+\abs{\Phi}2^{-B/30}$.
This is crucial to the efficiency of our algorithm.

\subsection{Analyze the main algorithm}\label{sec-main-analysis}
Now, we are ready to put everything together and analyze the whole algorithm.
Recall that we use $n$ to denote the number of vertices in the input graph and $m \geq n-1$ the number of edges.



We will need a simple lemma.

\begin{lemma}  \label{lem:conditional-dtv}
  Let $X$ be a random variable over some finite state space $\Omega$. Let $E \subseteq \Omega$ be an event that occurs with positive probability. Let $Y$ be the random variable $X$ conditional on $E$.
  Then,
  \begin{align*}
	\DTV{X}{Y} \leq \Pr[\overline{E}]. 
  \end{align*}
\end{lemma}

\begin{proof}
  We couple $X$ and $Y$ as follows: (1) first sample an indicator variable whether the event $E$ occurs; (2) if $E$ occurs, couple $X$ and $Y$ perfectly; and (3) if $E$ does not occur, independently sample $X$ conditional of $\overline{E}$ and sample $Y$. 
  By \Cref{lem:coupling},
  \begin{align*}
	\DTV{X}{Y} \leq \Pr[X \neq Y] &\leq \Pr[\overline{E}]. \qedhere
  \end{align*}
\end{proof}

The main goal of this section is to prove the following lemma.
%
In the next lemma, we consider a variant of \Cref{alg-main}, where we replace the subroutine \textsf{Sample} with the subroutine \textsf{NewSample} in \Cref{def-new-sample}.
\Cref{ob-new-sample} shows that \textsf{Sample} and \textsf{NewSample} have the same output distribution.
Hence, the replacement does not change the distributions of $\widetilde{R}_{v_i}$, $\widetilde{R}(V_0,E_0,\Lambda_0)$ and $S_{v_i}$ for all $1 \leq i \leq n$ and all $(V_0,E_0,\Lambda_0) \in \Phi_i$.
The only difference is that this variant of \Cref{alg-main} cannot be implemented in polynomial time.
However, we only use the variant algorithm to analyze the approximation error of \Cref{alg-main}.
The running time of \Cref{alg-main} is analyzed separately in the proof of \Cref{thm:main}.

\begin{lemma}\label{lemma-induction}
For any $1 \leq i \leq n$, there exists a good event $\+A(i)$ such that $\+A(i)$ occurs with probability at least $1 - \frac{n-i}{10n}$ and conditional on $A(i)$, it holds that 
\begin{itemize}
\item for any $j \geq i$, let $S^{\-{ideal}}_{v_j}$ be a multi-set of $\ell$ independent perfect samples from $\pi_{v_j}$, it holds that
\begin{align}\label{eq-2}
	\DTV{ (S_{v_j})_{j\geq i} }{(S^{\-{ideal}}_{v_j})_{j\geq i}} \leq 2^{-4m}(2^{n-i}-1);
	\end{align} 
  \item the following event, denoted by $\+B(i)$, occurs:
    \begin{align}\label{eqn:C_i}
      \forall(V_0,E_0,\Lambda_0) &\in \Phi_i, &
      1-\frac{n-i}{50n\max\{m,\eps^{-1}\}}  &\leq \frac{\widetilde{R}(V_0,E_0,\Lambda_0)}{R(V_0,E_0,\Lambda_0)} \leq  1+\frac{n-i}{50n\max\{m,\eps^{-1}\}}.
    \end{align}
\end{itemize}
In particular, the event $\+B(i)$ implies that, 
  \begin{align}\label{eq-1}
    1-\frac{n-i}{50n\max\{m,\eps^{-1}\}} \leq \frac{\widetilde{R}_{v_i}}{R_{v_i}} \leq 1+\frac{n-i}{50n\max\{m,\eps^{-1}\}}.
  \end{align}
\end{lemma}
\begin{proof}
We first show that $\+B(i)$ implies~\eqref{eq-1}. If $i = n$, then $\widetilde{R}_{v_i} = R_{v_i}$ and~\eqref{eq-1} holds.  For $i < n$, $\widetilde{R}_{v_i}$ is computed in \Cref{line-count} of \Cref{alg-main}, where the input $(V_0,E_0,\Lambda_0) \in \Phi_i$. Hence, $\+B(i)$ implies~\eqref{eq-1}.

Next, we construct the event $\+A(i)$ inductively from $i = n$ to $i = 1$ and prove ~\eqref{eq-2} and~\eqref{eqn:C_i}. 
The base case is $i = n$, where $v_n = t$. 
In this case, the only possible sample in $S_{v_n}$ is the empty graph with one vertex $t$, and thus~\eqref{eq-2} holds.
Also note that if $(V_0,E_0,\Lambda_0) \in \Phi_n$, then $\partial\Lambda_0 = \emptyset$. By \Cref{line-trival} in \Cref{alg-count}, 
if $t \in \Lambda$, the output is exactly $1$, and if $t \notin \Lambda$, the output is exactly $0$. In either case,~\eqref{eqn:C_i} holds.
Hence, we simply let $\+A(n)$ be the empty event, occurring with probability 1.

For $i<n$, we inductively define
\begin{align*}
  \+A(i) \defeq \+A(i+1) \land \+B(i) \land \+C(i),
\end{align*}
where the event $\+C(i)$ is defined next.
Recall that \Cref{alg-main} calls \textsf{Sample}  $\ell$ times to generate a multi-set of $\ell$ samples in $S_{v_i}$.
By \Cref{ob-new-sample}, \textsf{Sample} and \textsf{NewSample} (\Cref{def-new-sample}) have the same output distribution.
For analysis purposes, suppose we call \textsf{NewSample} instead $\ell$ times to generate a multi-set of $\ell$ samples.
%
%
In \eqref{eq-def-C}, we defined an event $\+C$ for one instance of $\textsf{NewSample}$. Since we have $\ell$ of them, define
\begin{align*}
	\+C(i): \text{for all } \ell \text{ calls of } \textsf{NewSample}, \text{ the event } \+C \text{ occurs.}
\end{align*}

%
%




Clearly $\+A(i)$ implies $\+B(i)$ and \eqref{eq-1}.
Before we show that $\+A(i)$ implies~\eqref{eq-2},
we first lower bound the probability of $\+A(i)$ by $1 - \frac{n-i}{10n}$.
By the induction hypothesis, since $\+A(i)$ implies $\+A(i+1)$, conditional on $\+A(i)$, we have
\begin{align}\label{eq-dtv-1}
	\DTV{ (S_{v_j})_{j > i} }{(S^{\-{ideal}}_{v_j})_{j> i}} \leq 2^{-4m}(2^{n-i-1}-1).
\end{align}
In fact $\+A(i+1)$ implies $\+A(j)$ for all $j\ge i+1$.
Thus, by \eqref{eq-1},
\begin{align}\label{eq-all-B}
  \forall j \geq i+1,\quad  1-\frac{n-i-1}{50n\max\{m,\eps^{-1}\}} \leq \frac{\widetilde{R}_{v_j}}{R_{v_j}} \leq 1+\frac{n-i-1}{50n\max\{m,\eps^{-1}\}},
\end{align}
as the worst case of the bound is when $j=i+1$. 
Note that $\Phi_j \subseteq \Phi_i$ for all $j < i$.
Conditional on $\+A(i+1)$, we know that~\eqref{eqn:C_i} (with parameter $i$) already holds for all $(V_0,E_0,\Lambda_0) \in \cup_{j > i}\Phi_j$.
We only need to show~\eqref{eqn:C_i} holds for all $\widetilde{R}(V_0,E_0,\Lambda_0)$ with $(V_0,E_0,\Lambda_0) \in \Phi_i \setminus \cup_{j > i}\Phi_j$.
The event $\+A(i+1)$ does not bias the inside independent randomness $\+D(V_0,E_0,\Lambda_0)$ in \Cref{alg-count} that generates $\widetilde{R}(V_0,E_0,\Lambda_0)$ for  $(V_0,E_0,\Lambda_0) \in \Phi_i \setminus \cup_{j > i}\Phi_j$.
Combining~\eqref{eq-dtv-1},~\eqref{eq-all-B} and \Cref{lemma-ac-gen}, with probability at least $1 - 2^{-4m}(2^{n-i-1}-1) - |\Phi_i|2^{-B/30}$, it holds that $\forall (V_0,E_0,\Lambda_0) \in \Phi_i \setminus \cup_{j > i}\Phi_j$,
\begin{align}\label{eqn:all-phi}
  1 - \frac{n-i-1}{50n\max\{m,\eps^{-1}\}}  - \frac{2}{\sqrt{\ell_2}}	\leq \frac{\widetilde{R}(V_0,E_0,\Lambda_0)}{R(V_0,E_0,\Lambda_0)} \leq 1 + \frac{n-i-1}{50n\max\{m,\eps^{-1}\}}  + \frac{2}{\sqrt{\ell_2}}.
\end{align}
Also recall that
$\ell_0 = \left\lceil 10^4 n^2\max\{m^2,\eps^{-2}\} \right\rceil$  and $B = 60n + 150m$.
We have
\begin{align}\label{eq-increase}
  \frac{n-i-1}{50n\max\{m,\eps^{-1}\}}  + \frac{2}{\sqrt{\ell_2}}	\leq \frac{n-i}{50n\max\{m,\eps^{-1}\}},
\end{align}
which means that \eqref{eqn:all-phi} implies $\+B(i)$.
By \Cref{lemma-ac-gen}, we have
\begin{align}
  \Pr[\+B(i) \mid \+A(i+1)] &\geq 1 - |\Phi_i|2^{-B/30} - 2^{-4m}(2^{n-i-1}-1)\label{eq-tmp}\\
  &\geq 1  - 2^{-4m} - 2^{-4m}(2^{n-i-1}-1) \geq 1 - 2^{-n},\label{eq-p-2}
\end{align}
where we used the fact that $|\Phi_i| \leq 2^{m + 2n}$.

Given $\+A(i+1) \land \+B(i)$, \eqref{eq-1} also holds.
The $\ell$ samples in $S_{v_i}$ are generated by \textsf{NewSample} on the graph~$G_{v_i}$.
In \Cref{line-c1} and \Cref{line-c2} of \Cref{alg-sample} (which are also in \textsf{NewSample}), \Cref{alg-count} is evoked to compute the value of $c_0$ and $c_1$. 
By \Cref{lem-input-condition}, the input $(V_0,E_0,\Lambda_0) \in \Phi_i$. 
Also, by~\eqref{eq-1}, $\widetilde{R}_{v_i}$ approximates ${R}_{v_i}$.
Hence the subroutine \Cref{alg-count} behaves like the oracle $\+P$ assumed in \Cref{lemma-new-sample}, satisfying conditions \eqref{eq-r1}, \eqref{eq-r2}, and \eqref{eq-r3}.
By the definition of $\+C(i)$ and \Cref{lemma-new-sample}, it is independent from $\+A(i+1) \land \+B(i)$ because $\+C(i)$ depends only on the independent randomness inside \textsf{NewSample}. By a union bound over $\ell$ calls of \textsf{NewSample}, we have
\begin{align}\label{eq-p-1}
  \Pr[\+C(i) \mid \+A(i+1) \land \+B(i)] \ge 1 - \frac{\eps^{200}\ell}{n^{200}} \geq 1 - \frac{1}{10n^2}, 
\end{align}
where $\ell = (60n + 150m)\left\lceil 10^5 n^3\max\{m^2,\eps^{-2}\}  \right\rceil$.
By the induction hypothesis,  \eqref{eq-p-2}, and \eqref{eq-p-1},
\begin{align*}
  \Pr[\+A(i)] &= \Pr[\+A(i+1) \land \+B(i) \land \+C(i)] \ge \left( 1 - \frac{n-i-1}{10n} \right)\left( 1 - \frac{1}{10n^2} \right)\left( 1 - 2^{-n} \right)\ge 1 - \frac{n-i}{10n}. 
\end{align*}

We still need to show that $\+A(i)$ implies~\eqref{eq-2}.
We use $(S_{v_j})_{j> i}|_{\+A(i+1)}$ to the denote the random samples of~$(S_{v_j})_{j> i}$ conditional on $\+A(i+1)$. 
Similarly, we can define $((\widetilde{R}_{v_j})_{j> i},(S_{v_j})_{j> i},\+D)|_{\+A(i+1)} $, where $\+D = (\+D(V_0,E_0,\Lambda_0))_{(V_0,E_0,\Lambda_0) \in \Phi_i}$,
and $((\widetilde{R}_{v_j})_{j\geq i},(S_{v_j})_{j\geq i},\+D)|_{\+A(i+1)\land \+B(i)} $, where we further condition on~$\+B(i)$. 
Note that the event $\+B(i)$ is determined by the random variables $((\widetilde{R}_{v_j})_{j> i},(S_{v_j})_{j> i},\+D)$.
By letting $E$ be $\+B(i)$ conditional on $\+A(i+1)$ in \Cref{lem:conditional-dtv},
we have that
\begin{align*}
  &\DTV{((\widetilde{R}_{v_j})_{j> i},(S_{v_j})_{j> i},\+D)|_{\+A(i+1)\land \+B(i)}}{((\widetilde{R}_{v_j})_{j> i},(S_{v_j})_{j> i},\+D)|_{\+A(i+1)}}\notag\\
  \leq \, & 1 - \Pr[\+B(i) \mid \+A(i+1)]\\
  \text{(by~\eqref{eq-tmp})}\,\,	\leq\,& |\Phi_i|2^{-B/30} + 2^{-4m}(2^{n-i-1}-1) \leq 2^{-4m} + 2^{-4m}(2^{n-i-1}-1).\notag
\end{align*}
Projecting to $(S_{v_j})_{j>i}$ we have
\begin{align*}
  \DTV{(S_{v_j})_{j> i}|_{\+A(i+1)\land \+B(i)}}{(S_{v_j})_{j> i}|_{\+A(i+1)}} \leq 2^{-4m} + 2^{-4m}(2^{n-i-1}-1).
\end{align*}
By the induction hypothesis, it holds that
\begin{align*}
 \DTV{(S_{v_j})_{j> i}|_{\+A(i+1)}} {(S_{v_j}^{\-{ideal}})_{j> i}} \leq 2^{-4m}(2^{n-i-1}-1).	
\end{align*}
Using the triangle inequality for total variation distances, we have
\begin{align*}
\DTV{(S_{v_j})_{j> i}|_{\+A(i+1)\land \+B(i)}} {(S_{v_j}^{\-{ideal}})_{j> i}} \leq 	2^{-4m} + 2^{-4m}(2^{n-i-1}-1) \times 2 = 2^{-4m}(2^{n-i}-1).
\end{align*}

Given $\+A(i+1) \land \+B(i)$, \eqref{eq-1} also holds.
The $\ell$ samples in $S_{v_i}$ are generated by \text{NewSample} on the graph~$G_{v_i}$.
By~\Cref{lemma-new-sample}, in that case, if $\+C(i)$ occurs, $S_{v_i}$ contains $\ell$ prefect independent samples. Furthermore, the event $\+C(i)$ is independent from $(S_{v_j})_{j> i}$ (as by~\Cref{lemma-new-sample}, $\+C(i)$ depends only on the internal independent randomness of \textsf{NewSample}
\footnote{In fact, $(S_{v_j})_{j> i}$ is correlated with $\+X_{\+P}$ in \Cref{lemma-new-sample} but independent from $\+D_u$.} ). 
Hence,
\begin{align*}
\DTV{ (S_{v_j})_{j\geq i}|_{\+A(i)} }{(S^{\-{ideal}}_{v_j})_{j\geq i}} &= 	\DTV{(S_{v_j})_{j> i}|_{\+A(i)}} {(S_{v_j}^{\-{ideal}})_{j> i}}\\
&= \DTV{(S_{v_j})_{j> i}|_{\+A(i+1)\land \+B(i)}} {(S_{v_j}^{\-{ideal}})_{j> i}}  \leq 2^{-4m}(2^{n-i}-1). \qedhere
\end{align*}
\end{proof}

We remark that the set $S_{v_i}$ may be used multiple times throughout \Cref{alg-main}.
In particular, this means that there may be subtle correlation among $\widetilde{R}_i$'s.
These correlations do not affect our approximation guarantee.
This is because the conditions of \Cref{lemma-ac} and \Cref{lemma-ac-gen} only involve marginals.
Namely, as long as the marginals are in the suitable range, the correlation amongst them does not matter.

By \eqref{eq-1} of \Cref{lemma-induction} with $i = 1$, note that $v_1 = s$, we have
\begin{align}\label{eq-Rs-1}
  \Pr\left[  1-\frac{1}{50\max\{m,\eps^{-1}\}} \leq \frac{\widetilde{R}_{s}}{R_{s}} \leq 1+\frac{1}{50\max\{m,\eps^{-1}\}} \right] \geq \Pr[\+A(1)] \geq \frac{3}{4}. 
\end{align}
Note that the events $\+A(i)$ for $1 \leq i \leq n$ are defined for the variant of \Cref{alg-main}, where we replace \textsf{Sample} with \textsf{NewSample}.
By~\Cref{ob-new-sample}, the variant and \Cref{alg-main} have the same output distribution. Hence, for the original \Cref{alg-main}, \eqref{eq-Rs-1} still holds. 

\begin{proof}[Proof of \Cref{thm:main}]
  The approximation guarantee follows directly from~\eqref{eq-Rs-1}. 
  Note that this guarantee is always at least a $(1\pm 1/m)$-approximation
  and is stronger than a $(1\pm\eps)$-approximation when $\eps>1/m$.
  We need this because in the analysis for the sampling subroutine we need to apply a union bound for the errors over the edges.

We analyze the running time next. Recall that $n$ is the number of vertices of the input graph and $m \geq n - 1$ is number of edges in $G$. 
Recall 
\begin{align*}
  \ell =  O((n+m)n^2\max\{m^2,\epsilon^{-2}\}) = O \tp{ n^2m \max\{m^2,\eps^{-2}\}}.
\end{align*}
By \Cref{lemma-ac}, the running time of \textsf{ApproxCount} (\Cref{alg-count}) is at most
\begin{align*}
  T_{\textsf{count}} = O(mn\ell) =  O \tp{ n^3m^2 \max\{m^2,\eps^{-2}\} }.
\end{align*}
By \Cref{lemma-sample}, the running time of \textsf{Sample} (\Cref{alg-sample}) is at most 
\begin{align*}
  T_{\textsf{sample}} = \widetilde{O}\tp{(n+m)T_{\textsf{count}}} = \widetilde{O}\tp{mT_{\textsf{count}}} = \widetilde{O}\tp{n^3m^3 \max\{m^2,\eps^{-2}\}}.
\end{align*}
Hence, the running time of \Cref{alg-main} is 
\begin{align*}
	T & = O\tp{n\tp{T_{\textsf{count}} + \ell T_{\textsf{sample}}}} = O\tp{n \ell T_{\textsf{sample}}} = \widetilde{O}\tp{n^6m^4\max\{m^4,\eps^{-4}\}}. \qedhere
\end{align*}
\end{proof}

\section{\BIS-hardness for \texorpdfstring{$s-t$}{s-t} unreliability}
\label{sec:BIS-hard}

In this section we show \Cref{thm:st-unrel}.
We first reduce \BIS{} to $s-t$ unreliability where each vertex (other than $s$ and $t$) fails with $1/2$ probability independently.
Note that in this version of the problem no edge would fail.
Given a DAG $D=(V\cup\{s,t\},A)$, this is equivalent to counting the number of subsets $S\subseteq V$ such that in the induced subgraph $D[S\cup\{s,t\}]$, $s$ cannot reach $t$.
We call $S$ a $\nrch{s}{t}{}$ set.

Given a bipartite graph $G=(V,E)$,
let its two partitions be $L$ and $R$.
We add two special vertices $s$ and $t$,
and connect, with directed edges, $s$ to all vertices in $L$ and all vertices in $R$ to $t$.
Lastly, for any edge $\{u,v\}\in E$, where $u\in L$ and $v\in R$, we replace it with a directed edge $(u,v)$.
Call the new directed graph $D_G$.
Clearly it is a DAG.

For any independent set $I$ in $G$, we claim that in $D_G[I\cup\{s,t\}]$, $s$ cannot reach $t$.
This is because for any $e\in E$, there is at least one vertex unoccupied.
Thus $s$ cannot reach $t$ using the directed version of $e$, and this holds for any $e\in E$.

In the other direction, let $S$ be a $\nrch{s}{t}{}$ subset of $V$.
This means that for any edge $\{u,v\}\in E$, either $u\not\in S$ or $v\not\in S$, as otherwise $s\rightarrow u\rightarrow v\rightarrow t$.
This means that $S$ is an independent set of $G$.

Thus, there is a one-to-one correspondence between independent sets of $G$ and $\nrch{s}{t}{}$ subsets of $V$.
Namely, $s-t$ unreliability where vertices (other than $s$ and $t$) fail with $1/2$ probability is \BIS{}-hard.

Next, we reduce $s-t$ unreliability from the vertex version.
For this, we can replace each vertex $v$ (other than $s$ and $t$) by two vertices $v,v'$ and a directed edge $v\rightarrow v'$.
All incoming edges of $v$ still goes into $v$, and all outgoing edges of $v$ goes out from $v'$.
Assign to the new edges the failure probabilities of their corresponding vertices,
and assign failure probability $0$ to all original edges.
Clearly the unreliability is the same with these changes.
To make failure probabilities uniform,
we can replace edges with failure probability $0$ by $k$ parallel edges.
Effectively, the connection fails only if all the parallel edges fail at the same time.
If the failure probability of each edge is $q$, the probability of all parallel edges failing is $q^k$.
As this probability approaches $0$ exponentially fast, it is easy to set a polynomially bounded $k$ so that the new unreliability is a sufficiently good approximation of the original.

As a side note, the last reduction also works for reliability.
Thus \Cref{thm:main} also works for $s-t$ reliability in DAGs where vertices rather than edges fail independently.

\ifdoubleblind
\else
\section*{Acknowledgement}
We thank Kuldeep S.~Meel for bringing the problem to our attention, Antoine Amarilli for explaining their method to us, and Marcelo Arenas for insightful discussions.
We also thank Zongchen Chen for suggesting a better presentation of \Cref{thm:main}, and Mark Jerrum for some useful insights.
This project has received funding from the European Research Council (ERC) under the European Union's Horizon 2020 research and innovation programme (grant agreement No.~947778).
Weiming Feng acknowledges the support from Dr. Max R\"ossler, the Walter Haefner Foundation and the ETH Z\"urich Foundation.
This work was done in part while Weiming Feng was visiting the Simons Institute for the Theory of Computing.
\fi

\bibliographystyle{alpha}
\bibliography{refs}

\appendix
\section{An alternative way to define the event \texorpdfstring{$\+C$}{C}} \label{app-proof}

We start from the following abstract setting. Let $A \sim \mu_A$ and $B \sim \mu_B$ be two random variables over some state space $\Omega$. 
Suppose for any $x \in \Omega$, it holds that 
\begin{align*}
  \mu_A(x) \geq (1 - \eps)\mu_B(x), 
\end{align*}
for some $0\le \eps<1$.
Then, the distribution $\mu_A$ can be rewritten as
\begin{align*}
	\mu_A = (1-\eps) \mu_B + \eps \nu,
\end{align*}
where the distribution $\nu$ is defined by
\begin{align*}
	\forall x \in \Omega, \quad \nu(x) = \frac{\mu_A(x) - (1- \eps) \mu_B(x) }{\eps}.
\end{align*}
Then, we can draw a sample $A \sim \mu_A$ using the following procedure.
\begin{itemize}
	\item Flip a coin with probability of HEADS being $1 - \eps$;
	\item If the outcome is HEADS, draw $A \sim \mu_B$;
	\item If the outcome is TAILS, draw  $A \sim \nu$.
\end{itemize}
In this procedure, we can define an event $\+C$ as ``the outcome of the coin flip is HEADS''. We know that conditional on $\+C$, the distribution of $A$ is $\mu_B$. Such an event $\+C$ is defined in an expanded space $\Omega \times \{\text{HEADS, TAILS}\}$.

Consider the modified version of \textsf{Sample}, where \textsf{Sample} is defined in \Cref{alg-sample}. Suppose we use it on $\pi_u$, where $u = v_k$. Let $\+X_{\+P}$ be the random variable associated with the oracle $\+P$.
Denote the distribution of $\+X_{\+P}$ by $\mu_{\+P}$.
\textsf{Sample} uses independent inside randomness $\+D_u$ to generate $H = H(\+X_{\+P}, \+D_u)$ and $F = F(\+X_{\+P}, \+D_u)$. Define
\begin{align*}
	\mu_A: \quad &\text{the joint distribution of } \+X_{\+P} \text{ and } H,\\
	\mu_B: \quad &\text{the product distribution of } \mu_{\+P} \text{ and } \pi_u.
\end{align*}
For any $x_{\+P}$ in the support of $\+X_{\+P}$ and any $h$ in the support of $\pi_u$, we have
\begin{align*}
	\frac{\mu_A(x_{\+P},h)}{\mu_B(x_{\+P},h)} = \frac{\Pr[\+X_{\+P} = x_{\+P}] \Pr[H =h \mid \+X_{\+P} = x_{\+P} ] }{\Pr[\+X_{\+P} = x_{\+P}] \pi_u(h)} = \frac{\Pr[H =h \mid \+X_{\+P} = x_{\+P} ] }{ \pi_u(h)}.
\end{align*}
Note that 
\begin{align*}
\Pr[H =h \mid \+X_{\+P} = x_{\+P} ] \geq \Pr[ F = 1 \mid \+X_{\+P} = x_{\+P} ] \Pr[ H =h \mid F = 1 \land \+X_{\+P} = x_{\+P} ]	 \geq (1-\eps)\pi_u(h),
\end{align*}
where $\eps = 1 - 1/n^{200}$.
 This implies 
 \begin{align*}
 	\frac{\mu_A(x_{\+P},h)}{\mu_B(x_{\+P},h)}  \geq 1 - \eps.
 \end{align*}
 Using the above abstract result, we can define an equivalent process of \textsf{Sample} and find an event $\+C$ such that $\Pr[\+C] \geq 1 - \eps$ and conditional on $\+C$, $(\+X_{\+P}, H) \sim \mu_B$. 
 By the definition of $\mu_B$, we know that conditional on $\+C$, $\+X_{\+P} \sim \mu_{\+P}$ still follows the distribution specified by the oracle $\+P$, which means that $\+C$ is independent from $\+X_{\+P}$. And also, $H$ is a perfect independent sample from $\pi_u$.
 Finally, we remark that in our analysis (see \Cref{ob-new-sample} and \Cref{lemma-induction}), we only need to show that such an equivalent process and the event $\+C$ exist. We do not need to implement the process nor certify the event in the algorithm. 

\section{Reducing counting \texorpdfstring{$s-t$}{st} connected subgraphs in DAGs to \texorpdfstring{\#NFA}{counting NFA}}
\label{sec:reduction}

Given a graph $G=(V,E)$ with a source $s$ and a sink $t$, the set of $s-t$ connected subgraphs are $\{H=(V,E_H)\mid E_H\subseteq E{~s.t.~}\rch{s}{t}{H}\}$.
Let $m\defeq\abs{E}$.
Counting $s-t$ connected subgraphs is equivalent to computing the $s-t$ reliability of the same graph with $q_e=1/2$ for all $e\in E$.
In this section, we reduce counting $s-t$ connected subgraphs in DAGs to \#NFA.
This reduction is essentially the same as the one in \cite{AvBM23},
where it is not explicitly given.

Given a DAG $G$, we construct an NFA $A_G$ such that the number of its accepting strings is the same as the number of $s-t$ connected subgraphs in $G$.
The states of $A_G$ consists of the starting state $s$, all edges, the accepting state $t$, a failure state, and some auxiliary states.
We order all edges in $G$ according to the head of the edge's topological order, say $e_1,\dots,e_m$. 
In particular, this means that if $f_1,\dots, f_k$ form a path, then $f_1\prec f_2 \dots \prec f_k$.
Moreover, we want to connect $s$ and $t$ to their respective adjacent edges, and two edges if they share an endpoint.
However, we want each bit of the input string to correspond to whether to have an edge or not,
which implies that we need to absorb all intermediate inputs.
Thus, instead, to connect $e_i$ to $e_j$ with $i<j$, we add auxiliary states $f_{k}^{(i,j)}$ from $k=i+1$ to $k=j-1$.
We connect $e_i$ to $f_{i+1}^{(i,j)}$, $f_{i+1}^{(i,j)}$ to $f_{i+2}^{(i,j)}$, etc., labelled with both $0$ and $1$.
Lastly, we connect $f_{j-1}^{(i,j)}$ to $e_j$, labelled with only $1$, and we connect $f_{j-1}^{(i,j)}$ to the failure state, labelled with $0$.
Once we are in the failure state, it can only move to itself, namely it has only a self-loop labelled with both $0$ and $1$.
We also do the same as above by treating $s$ as $e_0$ (whose tail is $s$ and head does not matter) and $t$ as $e_{m+1}$ (whose head is $t$ and tail does not matter).
Note that there are $O(m^2)$ states and we are counting accepting strings of length $m+1$.
The last bit of any accepting string has to be $1$, and therefore each accepting string is an indicator vector for a subset of edges.
It is easy to verify that the string is accepted if and only if $s$ can reach $t$ in the corresponding subgraph.
This finishes the reduction.

\end{document}